\documentclass[
        a4paper,
        authorcolumns,
        USenglish
    ]{lipics-v2019}

\usepackage[boxed, linesnumbered, noend, noline]{algorithm2e}
\SetAlCapFnt{\footnotesize}
\SetAlCapNameFnt{\footnotesize}
\SetInd{0.25em}{0.36em}
\SetAlCapSkip{1ex}
\SetAlFnt{\small}
\SetNlSkip{-0.3em}

\DontPrintSemicolon

\let\originalNl\nl                                                             
\newcommand\linesOn{\let\nl\originalNl}
\newcommand\linesOff{\let\nl\relax}

\usepackage[
        WOprobsubsuperscripts,                                                 
    ]{auxlib}

\def\whpFootnoteText/{
    We say a property of a protocol holds \emph{\whplongText/} if for each constant $a$, the constant parameters of the protocol can be set such that the property holds for each sufficiently large population size $n$ with probability at least $1 - n^{-a}$.
}

\crefname{prop}{Property}{Properties}
\Crefname{prop}{Property}{Properties}
\crefname{case}{Case}{Cases}
\crefname{guarantee}{Guarantee}{Guarantees}
\crefname{obs}{Observation}{Observations}
\crefname{reas}{Reason}{Reasons}
\crefname{type}{Type}{Types}

\Crefname{subsubsection}{Section}{Sections}

\usepackage[
        numbers,
        sort&compress
    ]{natbib}
\customtodo{peter}{Peter:}{color=AUXLorange}

\newcommand{\undervert}[2]{%
    \begin{array}[t]{@{}c@{}}
    #1 \\
    \stackrel{|}{\scriptstyle\text{#2}}
    \end{array}
}

\def\ProtocolBackupMaj/{\textsc{BackupMajority}}
\def\ProtocolFormJunta/{\textsc{FormJuntaUniform}}
\def\ProtocolFormJuntaExt/{\textsc{FormJunta}}

\newcommand*{\ProtocolSimpleMaj}[1]{\ensuremath{\textsc{SimpleMajority}_{#1}}}
\newcommand*{\ProtocolStableMaj}[1]{\ensuremath{\textsc{StableMajority}_{#1}}}
\newcommand*{\ProtocolConvergentMaj}[1]{\ensuremath{\textsc{ConvergentMajority}_{#1}}}
\newcommand*{\ProtocolUniformMaj}[1]{\ensuremath{\textsc{UniformMajority}_{#1}}}

\newcommand*{\ProtocolPhaseClock}[1]{\ensuremath{\textsc{PhaseClock}_{#1}}}

\newcommand*{\lmax}{l_{\max}}

\def\TimeC/{\ensuremath{T_{\operatorname{C}}}}
\def\TimeINF/{\ensuremath{T_{\operatorname{INF}}}}
\def\TimeLB/{\ensuremath{T_{\operatorname{LB}}}}
\def\TimeST/{\ensuremath{T_{\operatorname{ST}}}}

\def\bTrue/{\ensuremath{\textsc{true}}}
\def\bFalse/{\ensuremath{\textsc{false}}}

\newcommand*{\Round}{\mathrm{H}}
\newcommand*{\RoundStart}{\operatorname{R}_{\mathrm{Start}}}
\newcommand*{\RoundEnd}{\operatorname{R}_{\mathrm{End}}}
\newcommand*{\RoundLength}{\operatorname{R}_{\mathrm{Length}}}
\newcommand*{\RoundStretch}{\operatorname{R}_{\mathrm{Stretch}}}

\newcommand*{\BiasAbsolute}{\alpha}

\newcommand*{\TLoad}[1]{\Phi(#1)}
\newcommand*{\STLoad}[1]{\Psi(#1)}

\newcommand*{\LExpl}[1]{\operatorname{expl}_{#1}}

\def\chosenxi/{0.02}
\def\chosenoneminusxi/{0.98}

\newcommand*{\PCdef}[2]{%
    \expandafter\newcommand\csname PC#1\endcsname[1]{%
        \ensuremath{\text{\normalfont\texttt{PC#2}}(##1)}%
    }%
}

\PCdef{differentHour}{differentHour}
\PCdef{largerHour}{largerHour}
\PCdef{finishedJunta}{finishedJunta}
\PCdef{marked}{marked}
\PCdef{newHour}{newHour}
\PCdef{overflowed}{overflowed}
\PCdef{sameHour}{sameHour}
\PCdef{skippedHour}{skippedHour}
\PCdef{smallerHour}{smallerHour}

\newcommand*{\PCcall}[2]{\texttt{\ProtocolPhaseClock{#1}(\ensuremath{#2})}}    

\newcommand*{\Vardef}[2]{%
    \expandafter\newcommand\csname #1\endcsname[1]{%
        \ensuremath{\operatorname{\MakeLowercase{#2}}_{##1}}%
    }%
}

\Vardef{Count}{Count}
\Vardef{Error}{Error}
\Vardef{Finished}{Finished}
\Vardef{Load}{Load}

\title{%
    Time-space Trade-offs in Population Protocols
    for~the~Majority~Problem
}

\author{Petra Berenbrink}{Universität Hamburg, Germany}{petra.berenbrink@uni-hamburg.de}{}{}
\author{Robert Elsässer}{University of Salzburg, Austria}{elsa@cosy.sbg.ac.at}{}{}
\author{Tom Friedetzky}{Durham University, U.K.}{tom.friedetzky@dur.ac.uk}{}{}
\author{Dominik Kaaser}{Universität Hamburg, Germany}{dominik.kaaser@uni-hamburg.de}{}{}
\author{Peter Kling}{Universität Hamburg, Germany}{peter.kling@uni-hamburg.de}{https://orcid.org/0000-0003-0000-8689}{}
\author{Tomasz Radzik}{King's College London, U.K.}{tomasz.radzik@kcl.ac.uk}{}{Tomasz Radzik's work was supported by EPSRC grant EP/M005038/1, \enquote{Randomized algorithms for computer networks}.}

\authorrunning{P. Berenbrink, R. Elsässer, T. Friedetzky, D. Kaaser, P. Kling, and T. Radzik}

\Copyright{%
    Petra Berenbrink,
    Robert Elsässer,
    Tom Friedetzky,
    Dominik Kaaser,
    Peter Kling, and
    Tomasz Radzik
}

\nolinenumbers
\hideLIPIcs

\EventEditors{John Q. Open and Joan R. Access}
\EventNoEds{2}
\EventLongTitle{42nd Conference on Very Important Topics (CVIT 2016)}
\EventShortTitle{CVIT 2016}
\EventAcronym{CVIT}
\EventYear{2016}
\EventDate{December 24--27, 2016}
\EventLocation{Little Whinging, United Kingdom}
\EventLogo{}
\SeriesVolume{42}
\ArticleNo{23}

\begin{document}
\maketitle
\begin{abstract}
{Population protocols are a model for distributed computing that is focused on
simplicity and robustness. A system of $n$ identical agents (finite state
machines) performs a global task like electing a unique leader or determining
the majority opinion when each agent has one of two opinions. Agents
communicate in pairwise interactions with randomly assigned communication
partners. Quality is measured in two ways:
\begin{enumerate*}[label=, afterlabel=]
\item the number of interactions to complete the task and
\item the number of states per agent.
\end{enumerate*}
We present protocols for the majority problem that allow for a trade-off
between these two measures. Compared to the only other trade-off
result~\cite{DBLP:conf/podc/AlistarhGV15}, we improve the number of
interactions by almost a linear factor. Furthermore, our protocols can be made
uniform (working correctly without any information on the population size $n$),
yielding the first uniform majority protocols that stabilize in a subquadratic
number of interactions.
}
\end{abstract}

{\section{Introduction}%
\label{sec:introduction}

In this article we consider the \emph{majority} problem in the probabilistic \emph{population model}.
Majority is a fundamental problem in distributed computing.
There are $n$ different agents, each with one of two opinions, say $A$ and $B$ and the goal is to agree on the opinion with the larger support.
This problem occurs when all elements of a distributed system have to reach consensus on the value of some parameter which reflects the prevailing opinion what this value should be.
Because of its importance, the majority problem is frequently used as a case study in analysis and comparison of strengths and limitations of various models of distributed computing.

The population model was introduced by \textcite{DBLP:conf/podc/AngluinADFP04, DBLP:journals/dc/AngluinADFP06} as a model to explore the computational power of resource-limited, mobile agents.
Agents are modeled as finite-state machines.
In every step, a pair of agents is chosen uniformly at random, observe each other's state, and perform a deterministic state transition.
This is called an \emph{interaction}.
States are mapped to outputs by a problem-specific output function.
In the case of the majority problem, one can think of an agent's output as being $A$ or $B$, indicating which opinion the agent believes to be the majority.

The quality of a population protocol is measured in terms of the number of interactions (the runtime) and the number of states per agent required to \enquote{successfully compute} the desired output.
The number of interactions is sometimes expressed in \emph{parallel time}, which divides the number of interactions by $n$ to account for the inherent parallelism of the system.
In order to avoid confusion, we stick to the actual number of interactions throughout the article.

There are several definitions for what is conceived as a \enquote{successful computation}.
A typical requirement is that the system must, eventually, reach a state with correct output and which is \emph{stable} – i.e., no possible future transition can change the agents' output.
However, runtime notions differ in \emph{when} this strict guarantee must be achieved.
A natural definition is to measure the number $t$ of interactions after which the system is in such a stable state with correct output.
This notion is used in most recent publications, especially for lower bounds (cf.~\cref{sec:related_literature}).
Another definition considers the number of interactions $t$ after which the \emph{current execution} always gives the correct output.
The former runtime notion is typically referred to as \emph{stabilization} and the latter as \emph{convergence} (see~\cref{sec:model}).

One may wonder what the advantage in measuring the convergence time instead of the stabilization time may be.
In~\cite{DBLP:journals/dc/AngluinAE08a} the authors introduce a hybrid protocol that combines a \enquote{fast} protocol that might never converge to the correct answer with a \enquote{slow} one that stabilizes at the correct answer.
The hybrid protocol switches its output from the fast protocol, which might be incorrect, to that of the slow but always correct protocol when it is likely that the slow protocol has finished.
And therein lies the crux: without further safeguards, it is possible, although with only negligible probability, that a correct output reached by the fast protocol at time $t$ is later temporarily overwritten by a currently still wrong output of the slow protocol.
Hence, while the system has converged at time $t$, it is not yet stable.
It will stabilize only when the slow protocol does so.
The convergence (to the correct output) always happens by the time when the computation stabilizes (on the correct output).
The stabilization may, however, be reached later, sometimes much later, than convergence.

A desirable feature of population protocols is \emph{uniformity}, in the sense that a single algorithm should be designed to work for populations of any size.
Due to the simplicity of transition-based algorithms and the uniformity, uniform population protocols are well suited to model real-world systems that consist of many but comparatively simple agents, like a flock of birds or large sensor networks aggregating information (count, sum, average, extrema, median, or histogram).
In both scenarios the agents' computational power is bounded and the algorithms should not depend on the number of agents.

The underlying theme of this article is to exhibit trade-offs in population protocols between the running time and the required number of states, highlighting methods which help achieving fast stability (in addition to convergence) and uniformity of protocols.

\subsection{Our Contribution}%
\label{sec:our_contribution}

Our protocols for the majority problem in the population model provide an integer parameter $s \geq 2$ that enables a trade-off between the number of states and the runtime.
Our results also depend on the \emph{absolute bias} $\BiasAbsolute$, which is the initial absolute difference between the number of agents supporting opinion $A$ and $B$, respectively.
In the following we state the results for the tightest case when $\BiasAbsolute = 1$; see the corresponding theorems for the full statements.

Our first result is a comparatively simple protocol that, \whplong/, determines the exact majority in $\LdauOmicron{n \cdot {(\log n)}^2 / {\log s}}$ interactions and uses $\LDAUTheta{s + \log\log n}$ states (\cref{thm:majority_clocks}).
While this high-probability guarantee is comparatively weak with respect to the typical requirement of stabilization or even just convergence (since high-probability correctness allows for some low but positive probability of a permanent error), this protocol is an important building block for the following main results of this article.
\begin{enumerate}
\item

    We present two hybrid exact majority protocols, both having a runtime of $T = \LdauOmicron{n \cdot {(\log n)}^2 / {\log s}}$.
    One \emph{converges} \whplong/ in $T$ interactions and uses $\ldauTheta{s + \log\log n}$ states (\cref{thm:convergent-majority}).
    The other \emph{stabilizes} \whplong/ in $T$ interactions but uses $\LDAUTheta{s \cdot \log_s n}$ states (\cref{thm:majority_clocks_exact}).

\item

    For a constant $s$, we provide a uniform version of the second of the above two majority protocols.
    This protocol has essentially the same guarantee for the stabilization time.
    However, \whplong/ it uses $\LDAUOmicron{s \cdot \log_s n \cdot \log\log n}$ states (\cref{thm:majority:stabilizing:uniform}).

\end{enumerate}
All protocols above except for the uniform one need knowledge of $\floor{\log\log n}$.
Note that the state space of the uniform protocol is bounded only \whplong/; with negligible probability, an agent might need arbitrarily many states.
Since this is not covered in the original population model (where agents are \emph{finite}-state machines), for this protocol we adopt a generalized model~\cite{DBLP:conf/podc/DotyE19} in which agents are modeled as Turing machines (see \cref{sec:uniform-population-protocols}).

We highlight a few implications of the above results.
For a constant $s$, our majority results underline an important difference between stabilization and convergence.
While the $\LDAUTheta{\log n}$ number of states in our stable protocol (\cref{thm:majority_clocks_exact}) is asymptotically tight for any protocol that stabilizes \whplong/ in a subquadratic number of interactions\footnote{Conditioned on some natural properties satisfied by any known protocol, see~\cref{sec:related_literature}.}~\cite{DBLP:conf/soda/AlistarhAG18}, our protocol with $\LDAUTheta{\log\log n}$ states (\cref{thm:convergent-majority}) shows that the $\LDAUOmega{\log n}$  lower bound can be bypassed if one considers convergence instead of stabilization.

When choosing $s = \log\log n$, our majority protocols converge and stabilize \whplong/ in $\LdauOmicron{n \cdot {(\log n)}^2 / {\log\log\log n}}$ interactions.
These and the protocols presented in~\cite{DBLP:conf/wdag/BerenbrinkEFKKR18} are the first majority protocols with $\LDAUOmicron{\polylog n}$ states that work in $\Ldauomicron{n \cdot {(\log n)}^2}$ interactions.

When choosing $s = n^{\epsilon}$, where $\epsilon >0$ is an arbitrary positive constant, we obtain a majority protocol that stabilizes within asymptotically optimal $\LDAUOmicron{n \log n}$ interactions using $\LDAUTheta{n^{\epsilon}}$ states.
Before our work, achieving this optimal time required $\LDAUTheta{n^{3/2}}$ states~\cite{DBLP:conf/nca/MocquardAABS15}.

For a constant parameter $s$, our uniform protocol that stabilizes in $\LdauOmicron{n \cdot {(\log n)}^2}$ interactions and uses $\ldauOmicron{\log n \cdot \log\log n}$ states (\cref{thm:majority:stabilizing:uniform}) is the first uniform majority protocol that stabilizes in a subquadratic number of interactions, regardless of the required number of states.

An import ingredient for our results is an improvement to the \emph{phase clock} from~\cite{DBLP:conf/soda/GasieniecS18} – a distributed synchronization mechanism for population protocols.
Although this phase clock itself requires just a constant number of states, it is driven by a \emph{junta} of $n^{\epsilon}$ agents (for a constant $\epsilon \in \intco{0, 1}$), and selecting such a junta requires $\LDAUTheta{\log\log n}$ states.
By careful changes to the internals of the junta selection protocol and the interplay between the junta and the phase clocks, we not only simplify the phase clock protocol but also allow agents to \enquote{forget} some of the values required to select the junta.
This enables us to reduce the number of states required by our majority protocols from a \emph{factor} of $\LDAUTheta{\log\log n}$ to an \emph{additive} term of the same order.
See \cref{sec:phase_clock} for detailed explanations.

\subsection{Related Literature}%
\label{sec:related_literature}

The original population model was introduced by
\textcite{DBLP:conf/podc/AngluinADFP04, DBLP:journals/dc/AngluinADFP06},
assuming that the number of states per agent is constant. Together with
\textcite{AAE06, AAER07}, their results show that semilinear predicates (which
include, e.g., parity and majority) are stably computable in this model.
Subsequent results focused on quantifying the runtime and state requirements
for specific problems, in particular for the majority and the leader election
problems, and on generalizing the model. In the following overview we
concentrate on results in the population model for the majority problem. Bear
in mind that, as mentioned above, we state any runtime results in terms of the
required number of interactions, even when original sources state bounds in
parallel time only. For a broader overview of the extent of research and
results on protocols for the population model the reader is referred to the
survey papers~\cite{AspnesR2007} and~\cite{DBLP:journals/eatcs/ElsasserR18}.

\Textcite{DBLP:journals/dc/AngluinAE08} present a protocol with three states
and show that, \whplong/, the agents agree on the majority after
$\LDAUOmicron{n \log n}$ interactions if the initial difference between both
opinions (the \emph{absolute bias $\BiasAbsolute$}) is $\LDAUomega{\sqrt{n}
\log n}$. \Textcite{DBLP:conf/icalp/MertziosNRS14} show that, if agents are
required to succeed with probability $1$, at least four states are necessary.
They also provide a four state protocol that stabilizes \whplong/ in
$\LDAUOmicron{n^2 \log n}$ interactions. The same four state protocol was
independently (and earlier) studied by
\textcite{DBLP:journals/siamco/DraiefV12}, who proved similar results.
\Textcite{DBLP:conf/podc/AlistarhGV15} show a lower bound of $\LDAUOmega{n^2 /
\BiasAbsolute}$ on the expected interactions for any four state protocol. For
any number of states, they show a lower bound of $\LDAUOmega{n \log n}$
expected interactions.

To achieve fast runtime, \Textcite{DBLP:conf/nca/MocquardAABS15} consider the
population model allowing a super-constant number of states per agent. They
present a protocol that calculates the signed difference between the two
opinions' support \whplong/ in asymptotically optimal $\LDAUOmicron{n \log n}$
interactions but uses polynomial $\LDAUTheta{n^{3/2}}$ number of states. The
constant-state but slow quadratic-time
protocols~\cite{DBLP:journals/siamco/DraiefV12,DBLP:conf/icalp/MertziosNRS14}
on the one hand and the fast but polynomial-state
protocol~\cite{DBLP:conf/nca/MocquardAABS15} on the other, posed the quite
natural question of designing fast $\LDAUOmicron{n \polylog n}$-time majority
protocols which use a relatively small $\LDAUOmicron{\polylog n}$ number of
states.

\Textcite{DBLP:conf/soda/AlistarhAEGR17} show a lower bound on the required
number of interactions for population protocols with a small number of states.
For majority, their bound states that protocols with less than $(\log\log n)/2$
states require $\LDAUOmega{n^2 / \polylog(n)}$ interactions in expectation in
order to stabilize. \Textcite{DBLP:conf/soda/AlistarhAG18} further improve this
lower bound, by showing that any protocol that solves majority and stabilizes
in $n^{2 - \LDAUOmega{1}}$ expected interactions requires $\LDAUOmega{\log n}$
states. Both these lower bounds require certain natural monotonicity
assumptions which are satisfied by all known majority protocols.

A recent series of papers~\cite{DBLP:conf/soda/AlistarhAEGR17,
DBLP:conf/soda/AlistarhAG18, DBLP:conf/podc/AlistarhGV15,
DBLP:conf/podc/BilkeCER17, DBLP:conf/wdag/BerenbrinkEFKKR18} showed upper
bounds. \Textcite{DBLP:conf/soda/AlistarhAG18} present a protocol that
stabilizes \whplong/ in $\LdauOmicron{n \cdot {(\log n)}^2}$ interactions
and requires $\LDAUOmicron{\log n}$ states. In a recently published
result~\cite{DBLP:conf/wdag/BerenbrinkEFKKR18}, we present a population
protocol for majority that reduces the number of interactions to
$\LdauOmicron{n \cdot {(\log n)}^{5/3}}$, both in expectation and
\whplong/.

The subquadratic-time protocols for majority presented
in~\cite{DBLP:conf/soda/AlistarhAEGR17, DBLP:conf/soda/AlistarhAG18,
DBLP:conf/podc/AlistarhGV15, DBLP:conf/podc/BilkeCER17,
DBLP:conf/wdag/BerenbrinkEFKKR18, DBLP:conf/nca/MocquardAABS15} are not
uniform. To work correctly, they need an estimate of the size of the
population; more precisely, they need a value which is $\LDAUTheta{\log n}$.
They also, with exception of protocols proposed
in~\cite{DBLP:conf/podc/AlistarhGV15}, provide no means to trade runtime for
the number of states required per agent, as our protocols do.
\textcite{DBLP:conf/podc/AlistarhGV15} is the only paper we know of which
presents majority protocols with a trade-off of similar nature. For a parameter
$m \leq n$, their algorithm uses $\LDAUOmicron{m+ \log n \cdot \log m}$ states
and stabilizes \whplong/ in $\LdauOmicron{n^2 \cdot (\log n) / (\BiasAbsolute
\cdot m) + n \cdot {(\log n)}^2}$ interactions.

In parallel to our work, \textcite{DBLP:conf/podc/KosowskiU18} recently
designed population protocols, including two majority protocols that converge
in $\LdauOmicron{n{(\log n)}^3}$ and $\LDAUOmicron{n^{1 + \epsilon}}$
interactions and use $\LDAUOmicron{\log \log n}$ and constant $f(\epsilon)$
number of states, respectively. Here, $\epsilon$ is an arbitrarily chosen
positive constant.

With the only exception of~\cite{DBLP:journals/dc/AngluinAE08}, all majority protocols mentioned above solve exact majority.
That is, they eventually output the correct majority opinion with probability $1$.
This holds even if the initial bias towards one opinion is as small as only $1$.
}
{\section{Model \& Notation}%
\label{sec:model}

Population protocols are a computational model for a distributed system consisting of $n$ agents, in the following also referred to as \emph{nodes}.
Nodes are assumed to be identical finite-state machines\footnote{%
    For our uniform protocol, \cref{sec:uniform-population-protocols} introduces a generalized model where agents are Turing machines.
}.
In each time step, an ordered pair of nodes $(u, v)$ is chosen independently and uniformly at random.
Node $u$ is called the \emph{initiator} and node $v$ is called the \emph{responder}.
Let $s_u$ be the state of $u$ and $s_v$ be the state of $v$ at the beginning of such an \emph{interaction}.
Both nodes observe each other's state and update themselves according to a fixed, deterministic transition function of the form $(s_u, s_v) \mapsto (s_u', s_v')$.
At any time, the global state of the system can be fully described by a function $c$ that maps each node to its current state.
This function $c$ is called the \emph{configuration} of the system at that time.

Nodes try to reach and stay in a set of \emph{target configurations}, whose definition depends on the considered problem.
It is \emph{not} required, indeed not possible in this model, that nodes realize when a target configuration has been reached.
Target configurations are specified via an \emph{output function} of the form $s \mapsto o$ that maps a state $s$ to a (problem specific) output value $o$.

We are interested in population protocols for the majority problem, where nodes start in one of two different states (also called \emph{opinions}).
We seek a configuration in which all nodes agree on the opinion with the initially larger support.
The \emph{absolute bias} $\BiasAbsolute$ is the absolute difference between the initial number of supporters for each opinion.
We assume $\BiasAbsolute \geq 1$.
The output function maps each state $s$ to an output $o \in \set{+1, -1}$, representing one of the two opinions.
The target configurations are all configurations in which node states map all to $+1$, if $+1$ represents the initial majority opinion, or map all to $-1$, if $-1$ represents the initial majority.

The quality of a protocol is measured in terms of the number of interactions and the number of states per node required to reach and stay in target configurations.
There are two common ways to formalize what exactly is meant by \enquote{reach and stay}: \emph{stabilization time} and \emph{convergence time}.\footnote{%
    The notions as defined here are the ones used predominantly in population protocols in recent literature.
    However, note that some previous publications (e.g.,~\cite{DBLP:conf/soda/AlistarhAEGR17, DBLP:conf/podc/BilkeCER17}) refer to stabilization time as convergence time.
}
\begin{description}
\item[Convergence Time:]

    The \emph{convergence time} $\TimeC/$ of a protocol is the random variable that measures the number of interactions until the protocol has reached and remains in the set of target configurations.

\item[Stabilization Time:]

    We say a configuration $c$ is \emph{stable}, if in any configuration $c'$ that is reachable from $c$ by a sequence of interactions, each node has the same output as in $c$.
    The \emph{stabilization time} $\TimeST/$ of a protocol is the random variable that measures the number of interactions until the protocol has reached a stable target configuration.

\end{description}
Clearly, $\TimeC/ \leq \TimeST/$, since reaching a stable target configuration
implies that, whatever future interactions may be, the system will always
remain in a target configuration. The stabilization time $\TimeST/$ can,
however, be strictly larger than the convergence time $\TimeC/$.

As bounds on the convergence and stabilization time are given in probabilistic
terms, one often additionally emphasizes whether a protocol is guaranteed to,
eventually, reach a stable target configuration (i.e., whether $\TimeST/<
\infty$ holds with probability $1$). Such protocols are called \emph{exact} or
\emph{always correct}.

The newer results on population protocols, for example~\cite{DBLP:conf/soda/GasieniecS18, DBLP:conf/soda/AlistarhAG18}, tend to consider the stabilization time for exact protocols.
However, from a practical point of view, convergence may provide similarly strong runtime guarantees while enabling more efficient protocols.
Indeed, our \cref{thm:convergent-majority} shows that the lower bound on the number of states required by any majority protocol that \emph{stabilizes} in $n^{2 - \LDAUOmega{1}}$ expected interactions does not apply if one considers \emph{convergence} instead.

In the remainder of this article, we define $\N$ as the set of natural numbers
without zero and $\N_0 \coloneqq \N \cup \set{0}$.
}
{\section{Auxiliary Population Protocols}%
\label{sec:auxiliary}

In this section we introduce a few auxiliary population protocols that we use as subroutines.
These protocols, or variants of them, are well known and have been used in other work on population protocols, as indicated below.

We start with two comparatively simple primitives: \emph{One-way Epidemic} and \emph{Load Balancing}.
Afterward we proceed to describe two more involved protocols, one for the creation of a \emph{junta} (\cref{sec:junta_creation}) and one for the creation of a \emph{phase clock} (\cref{sec:phase_clock}), both of which require slight adaptions and rephrasing to fit into our setting.

\paragraph{One-way Epidemic}
A \emph{one-way epidemic} for $n$ nodes is a population protocol with state space $\set{0, 1}$ and transition function $(x, y) \mapsto (x, \max\set{x, y})$.
Nodes with value $0$ are referred to as \emph{susceptible} and nodes with value $1$ as \emph{infected}.
We define the \emph{infection time} \TimeINF/ as the number of interactions required by a one-way epidemic starting with a single infected node to infect the whole population.
The following upper and lower high-probability bounds on \TimeINF/ have been shown in~\cite{DBLP:journals/dc/AngluinAE08a}.
\begin{lemma}[{\cite[Lemma~2]{DBLP:journals/dc/AngluinAE08a}}]%
\label{lem:epidemic}
For any constant $a > 0$ there exist constants $c_1, c_2 > 0$ such that we have
the inequality
\begin{math}
     \Prob{c_1 \cdot n \log n \leq \TimeINF/ \leq c_2 \cdot n \log n}
\geq 1 - n^{-a}
\end{math}.
\end{lemma}

\paragraph{Load Balancing}
We define a simple population protocol for \emph{load balancing} over $n$ nodes.
The state space is $\set{-\Lambda, -(\Lambda-1), \dots, \Lambda-1, \Lambda}$, where $\Lambda \in \N$ is a positive integer (which may depend on $n$).
We say a node in state $x$ has \emph{load} $x$.
The transition function is $(x, y) \mapsto \bigl( \ceil*{\frac{x+y}{2}}, \floor*{\frac{x+y}{2}} \bigr)$.
Let $\Delta(t)$ denote the \emph{discrepancy} after $t$ interactions, which is the difference between the maximum and minimum load among all nodes, and set $\Delta \coloneqq \Delta(0)$.
We define the \emph{load balancing time} \TimeLB/ as the number of interactions required to reduce the initial discrepancy to at most $2$.
The following upper high-probability bound on \TimeLB/ has been shown in~\cite{journals/corr/BerenbrinkFKK18}.
\begin{lemma}[{\cite[from Theorem~1]{journals/corr/BerenbrinkFKK18}}]%
\label{lem:loadbalancing}
For any constant $a > 0$, there exists a constant $c > 0$ such that we have the
inequality
\begin{math}
     \Prob{\TimeLB/ \leq c \cdot n \log(n \cdot \Delta)}
\geq 1 - n^{-a}
\end{math}.
\end{lemma}
}
{\subsection{Junta}%
\label{sec:junta_creation}

The next protocol rapidly elects a non-empty junta of size at most $n^{1 -
\LDAUOmega{1}}$. It is a variant of a protocol
from~\cite{DBLP:conf/soda/GasieniecS18}, where each node calculates a
\emph{level} from a range of $\LDAUTheta{\log\log n}$ values and all nodes with
the highest level form the junta. The original protocol is modified such that
we can not only provide an upper bound on the highest level reached by any node
(as in~\cite{DBLP:conf/soda/GasieniecS18}) but also a lower bound. This
change also simplifies the protocol and allows the nodes to realize when the
junta selection has finished. Thus, in contrast
to~\cite{DBLP:conf/soda/GasieniecS18}, nodes are not required to store their
level ad infinitum. This is important when using the junta selection as a
subprotocol, as storing the level would then increase the number of states per
node by a \emph{factor} of $\LDAUTheta{\log\log n}$.

We first describe in \cref{sec:level-calculation} how the levels are
calculated. Then we continue to describe how this level calculation can be used
to calculate a junta with the desired properties and state the main result for
the junta election process in \cref{sec:junta-calculation}. The protocol's
analysis is given in \cref{sec:junta-auxiliary, sec:junta-analysis}.

\subsubsection{Level Calculation}\label{sec:level-calculation}

For the level calculation, the state of a node is a tuple of the form $(l, a)$,
where the \emph{level} $l \in \N_0$ is a counter and the \emph{activity} bit $a
\in \set{0, 1}$ indicates whether a node is active or not. Initially, all nodes
have state $(l, a) = (0, 1)$. That is, they are at level $0$ and active.

To describe the transition function, we distinguish between a node's first
interaction and any of its subsequent interactions. During its first
interaction, a node $u$ adopts state $(1, 1)$ if it is the initiator and state
$(0, 0)$ if it is the responder. Since the interacting nodes are chosen
randomly, this simulates a fair coin toss to decide whether the node should
remain active and move up to level $1$, or whether it should become inactive.

During any following interaction, $u$ changes its state only if it is still
active ($a = 1$) and if it is the initiator of the interaction. In this case,
when interacting with a responder in state $(l', a')$, node $u$ updates its
state as follows:
\begin{equation}%
\label{eq:level-calculation}
\bigl[ (l, 1), (l', a') \bigr] \mapsto
\begin{cases}
    (l + 1, 1)
&   \text{if $l' \geq l$ and}
\\
    (\phantomas[l]{l+1}{l}, 0)
&   \text{otherwise.}
\end{cases}
\end{equation}
In words, a node remains active and increases its level as long as it does not
encounter a node with a lower level. The only difference to the protocol
from~\cite{DBLP:conf/soda/GasieniecS18} is how nodes behave in their first
interaction, which allows us to provide a lower bound on the maximum level
reached by any node (\cref{lem:junta}). We use the random variable $L^*$ to
denote this maximum level. Moreover, for $l \in \N_0$ we define $B_l$ as the
number of nodes that reach level at least $l$ before becoming inactive.

\subsubsection{Junta Calculation}%
\label{sec:junta-calculation}

We now describe how the above level calculation can be used to calculate a suitable junta.
In addition to the level $l$ and activity bit $a$, each node stores a \emph{marker bit} $b \in \set{0, 1}$ that indicates whether the node is a member of the junta ($b = 1$) or not ($b = 0$) and a \emph{finished bit} $f \in \set{0, 1}$ that indicates whether a node knows that there is at least one marked node ($f = 1$) or not ($f = 0$).
Initially, all nodes have $b = 0$ and $f = 0$.
If two nodes with finished bit $0$ interact, they update their levels as described in~\cref{sec:level-calculation}.
Any node that reaches level $\lmax \coloneqq \floor{\log\log n} - 3$ sets its marker bit $b = 1$ and its finished bit $f = 1$.
If two nodes interact and at least one of them has its finished bit set to $1$, both nodes set their finished bit to $1$; no further state updates happen in this case.

We refer to this protocol as \ProtocolFormJuntaExt/.
An important difference to the junta protocol from~\cite{DBLP:conf/soda/GasieniecS18} is under which circumstances a node is assumed to be part of the junta.
While our protocol starts with an empty junta and marks a node as part of the junta when it reaches level $\lmax$, the protocol from~\cite{DBLP:conf/soda/GasieniecS18} assumes that a node is in the junta as long as it has not encountered a node with a higher level.
In particular, initially the junta from~\cite{DBLP:conf/soda/GasieniecS18} has linear size and decreases gradually over time.
Protocols using a junta typically rely on a junta of size at most $n^{1 - \LDAUOmega{1}}$.
Dealing with the initially oversize junta requires some care, a difficulty avoided by our protocol.
Another benefit of our protocol is that once a node sets its finished bit, its level value (and activity bit) are no longer of any relevance and need not be stored any longer.
These benefits come with the caveat that our protocol may not finish with a non-zero (but, as we will show, negligible) probability.
That is, it is possible that no node is ever marked/finished.

The remainder of this section proves the following \lcnamecref{thm:junta}.
\begin{theorem}%
\label{thm:junta}
Fix any constant $a > 0$ and let $n$ be sufficiently large with respect to $a$.
With probability at least $1 - n^{-a}$, protocol \ProtocolFormJuntaExt/ calculates a non-empty junta (with all nodes being finished) of size at most $n^{\chosenoneminusxi/}$ within $\LDAUOmicron{n \log n}$ interactions.
It uses $2 \cdot (\lmax + 1) = \LDAUTheta{\log\log n}$ states per node.
Finished nodes are in one of exactly two states, indicating whether the node is in the junta or not.
\end{theorem}
Note that our analysis of \cref{thm:junta} is not designed to keep the involved constants small but instead to make the asymptotic analysis as clear as possible.
For example, the current, simple asymptotic analysis wold require an exorbitant large value for $n$ ($\geq e^{800}$).
These numbers arise simply out of convenient choices and it is not difficult (if tedious) to improve them to more realistic values.
In fact, simple experimental simulations show that these protocols work already well in practice for values of $n \geq 10^6$.

\subsubsection{Auxiliary Claims about the Level Calculation}%
\label{sec:junta-auxiliary}

In this section we state and prove some auxiliary \lcnamecrefs{clm:junta:aux1} about the level calculation described in \cref{sec:level-calculation}.
We start with upper and lower bounds on the number $B_1$ of nodes that proceed from level $0$ to level $1$ (\cref{clm:junta:aux1}).
Afterward, we provide both upper and lower bounds on the number of nodes that proceed from level $l$ to level $l + 1$ for $l \in \N$ (\cref{clm:junta:aux2}).
Finally, we bound how many levels nodes can proceed beyond any level that is reached by at most $\LDAUOmicron{n^{1/3}}$ nodes (\cref{clm:junta:aux3}).
\begin{claim}%
\label{clm:junta:aux1}
Fix any two constants $a, \epsilon > 0$ and let $n$ be sufficiently large with respect to $a$ and $\epsilon$.
Then,
\begin{math}
     \Prob{\abs{B_1 - n/2} < \epsilon \cdot n/2}
\geq 1 - n^{-a}
\end{math}.
\end{claim}
\begin{proof}
For a node $u$ let the \emph{first interaction} $t_u$ of $u$ denote the earliest interaction during which $u$ was either initiator or responder.
We say $u$ is a \emph{singleton} if $t_u \neq t_v$ for all nodes $v \neq u$.
Two nodes $u \neq v$ with $t_u = t_v$ are called twins.
Let $\cS$ denote the set of all singletons and $\cT$ the set of all nodes that have a twin.

For each node $u$ we define the binary random variable $X_u$ to be $1$ if and only if $u$ is the initiator of $t_u$.
Note that $\Prob{X_u = 1} = 1/2$ and that $B_1 = \sum_{u} X_u$.
We would like to treat $B_1$ as a binomial distribution $\BinDistr(n, 1/2)$.
Unfortunately, the variables $X_u$ are not independent: for twins $u$ and $v$, exactly one of $X_u$ and $X_v$ is $1$.
To fix this, define $K \in \set{1, 2, \dots, \floor{n/2}}$ as the number of pairs $u$ and $v$ that are twins and let us condition on a fixed $K = k$.
The $n - 2k$ variables $X_u$ with $u \in \cS$ are completely independent of the remaining process (a node becomes initiator or responder independently with probability $1/2$).
For the $2k$ variables corresponding to twins, note that their sum is exactly $k$.
We pick an arbitrary subset $\cT_1 \subseteq \cT$ of $k$ twins and define $X'_u \coloneqq 1$ for all $u \in \cT_1$ as well as $X'_u \coloneqq 0$ for all $u \in \cT \setminus \cT_1$.
For $u \in \cS$, we define $X'_u \coloneqq X_u$.
Obviously, we have $B_1 = \sum_{u} X_u = \sum_{u} X'_u$ and the set of all $X'_u$ is mutually independent.
Moreover, $\Ex{B_1 | K = k} = k \cdot 1 + k \cdot 0 + (n - 2k)/2 = n/2$.
For any constant $b > 0$, Chernoff (\cref{eq:chernoff:mult:combined}) gives
\begin{equation*}
     \Prob{\abs{B_1 - n/2} \geq \delta \cdot n/2 | K = k}
\leq 2n^{-b}
,
\end{equation*}
where $\delta \coloneqq \sqrt{6b \cdot \log(n) / n} = \LDAUomicron{1}$.
Using the law of total probability to get rid of the conditioning yields the inequality
\begin{math}
     \Prob{\abs*{B_1 - n/2} \leq \epsilon \cdot n/2}
\geq 1 - 2n^{-b}
\end{math},
which implies the \lcnamecref{clm:junta:aux1}'s statement by choosing the constant $b = a + 1$.
\end{proof}

\begin{claim}\label{clm:junta:aux2}
Fix any two constants $a > 0$ and $\epsilon \in \intoc{0, 1}$ and let $n$ be sufficiently large with respect to $a$ and $\epsilon$.
For all $l \in \N$, $\xi_U \in \intco{n^{-1/3}, 1}$, and $\xi_L \in
\intco{n^{-1/2} \ln n, 1}$ we have
\begin{enumerate}[nosep]
\item\label{clm:junta:aux2:a}
    \begin{math}
         \Prob{B_{l+1} < (1 + \epsilon) \xi_U^2 \cdot \phantomas{n/4}{n} | B_l \leq \xi_U \cdot n}
    \geq 1 - n^{-a}
    \end{math}
    and

\item\label{clm:junta:aux2:b}
    \begin{math}
         \Prob{B_{l+1} > (1 - \epsilon) \xi_L^2 \cdot n/4 | B_l \geq \xi_L \cdot n}
    \geq 1 - n^{-a}
    \end{math}.

\end{enumerate}
\end{claim}
\begin{proof}
Fix an $l \in \N$ and consider a node $u$ that just reached level $l$. Node
$u$ is still active and will either become inactive or proceed to level $l
+ 1$ during its next interaction. Let $t$ be $u$'s next interaction.
\begin{enumerate}[nosep, wide]
\item

     The probability for $u$ to proceed to level $l + 1$ during interaction $t$
     is at most $B_l / n$. This holds for all $B_l$ nodes that reach level at
     least $l$. By a straightforward coupling argument\footnote{
        Run the original process and mark all nodes that reach level $l$. Then
        run the coupled process and use the same random choices. Proceeding
        from level $l'$ to $l' + 1$ for $l' \in \N_0 \setminus \set{l}$ works
        as in the original process. However, for a node to proceed from level
        $l$ to $l + 1$ its interaction partner must have been marked in the
        original process.
    }, we get that $B_{l+1}$ is stochastically dominated by a binomially
    distributed random variable $\BinDistr(B_l, B_l/n)$. Conditioned on $B_l
    \leq \xi \cdot n$ we can apply Chernoff (\cref{eq:chernoff:mult:b}) to get
    \begin{equation}
    \begin{aligned}
    &
    \Prob{B_{l+1} \geq (1 + \epsilon) \cdot \xi^2 \cdot n | B_l \leq \xi \cdot n}
    \\{}\leq{}&
    e^{-\frac{\epsilon^2 \cdot \xi^2 \cdot n}{3}}
    \leq
    e^{-\frac{\epsilon^2 \cdot n^{1/3}}{3}}
    ,
    \end{aligned}
    \end{equation}
    implying the desired statement.

\item

    If $u$ is among the last $B_l/2$ nodes that try to proceed from level $l$
    to level $l + 1$, its probability to proceed to level $l+1$ is at least
    $B_l / (2n)$. By a straightforward coupling argument\footnote{
        Run the original process and let $b$ denote the number of nodes that
        reach level $l$. Mark the first $b/2$ nodes that try to proceed from
        level $l$ to level $l + 1$. Then run the coupled process and use the
        same random choices. Proceeding to the next level works as in the
        original process, except for the last $b/2$ nodes that try to proceed
        from level $l$ to level $l + 1$: such nodes proceed only if their
        interaction partner has been marked in the original process.
    } shows that $B_{l+1}$ stochastically dominates a binomially distributed
    random variable $\BinDistr(B_l/2, B_l/(2n))$. Conditioned on $B_l \geq \xi
    \cdot n$ we can apply Chernoff (\cref{eq:chernoff:mult:a}) to get
    \begin{equation}
    \begin{aligned}
    &
    \Prob{B_{l+1} \leq (1 - \epsilon) \cdot \xi^2 \cdot n/4 | B_l \geq \xi \cdot n}
    \\{}\leq{}&
    e^{-\frac{\epsilon^2 \cdot \xi^2 \cdot n/4}{2}}
    \leq
    e^{-\frac{\epsilon^2 \cdot {(\ln n)}^2}{8}}
    ,
    \end{aligned}
    \end{equation}
    implying the desired statement.
    \qedhere

\end{enumerate}
\end{proof}

\begin{claim}%
\label{clm:junta:aux3}
Fix any integer constant $a \ge 1$ and let $n$ be sufficiently large.
For all $l \in \N$, we have
\begin{equation}
     \Prob{B_{l + 4a} = 0 | B_l < 2n^{1/3}}
\geq 1 - n^{-a}
.
\end{equation}
\end{claim}
\begin{proof}
Note that $B_l < 2n^{1/3}$ implies $B_{l'} \leq B_l < 2n^{1/3}$ for all $l'
\geq l$. By Markov's inequality, we have
\begin{equation}
\begin{aligned}
&
\Prob{B_{l' + 1} \geq 1 | B_{l'} < 2n^{1/3}}
\\{}\leq{}&
\Ex{B_{l'+1} | B_{l'} < 2n^{1/3}}
\leq
4n^{-1/3}
.
\end{aligned}
\end{equation}
We apply Markov's inequality to the next $4a$ levels and get
\begin{math}
     \Prob{B_{l + 4a} \geq 1 | B_l < 2n^{1/3}}
\leq {(4n^{-1/3})}^{4a}
\leq n^{-a}
\end{math}.
\end{proof}

\subsubsection{Analysis of the Junta Calculation}\label{sec:junta-analysis}

Equipped with the auxiliary \lcnamecrefs{clm:junta:aux1} from
\cref{sec:junta-auxiliary}, we continue with the analysis of the junta
calculation. First, we bound the time it takes until all nodes become inactive
(\cref{lem:junta:runtime}). Next, we give upper and lower bounds on the maximum
level $L^*$ reached by the nodes (\cref{lem:junta}) as well as an upper bound
on the number $B_{\lmax}$ of nodes that reach level $\lmax$
(\cref{lem:junta:nodesatlowerbound}). Finally, the proof of \cref{thm:junta} is
given at the end of this \lcnamecref{sec:junta-analysis}.
\begin{lemma}%
\label{lem:junta:runtime}
Fix any constant $a > 0$ and let $n$ be sufficiently large with respect to $a$.
With probability at least $1 - n^{-a}$ all nodes become inactive during the first $(6a + 12) \cdot n \ln n$ interactions.
\end{lemma}
\begin{proof}
The probability that a given node does not interact in a given interaction is
$1 - 1/n$. Thus, the probability that a given node does not interact at all
during the first $c \cdot n \ln n$ interactions is at most ${(1 - 1/n)}^{c
\cdot n \ln n} \leq n^{-c}$ for any $c > 0$. By a union bound, we get that all
nodes interacted at least once after the first $c \cdot n \ln n$ interactions
with probability at least $1 - n^{-c + 1}$. Together with \cref{clm:junta:aux1}
and a union bound, we know that, with probability $1 - 2n^{-c + 1}$, there are
at least $n/3$ nodes in state $(0, 0)$ after $c \cdot n \ln n$ interactions.

From that point on, the probability for any fixed node to become inactive
during a given interaction is at least $\frac{1}{3n}$ (the node is chosen as
the initiator of the interaction and its communication partner is one of the
$n/3$ nodes in state $(0, 0)$). Thus, the probability that any fixed node
remains active during the next $c \cdot n \ln n$ interactions is at most ${(1 -
1/(3n))}^{c \cdot n \ln n} \leq n^{-c/3}$. By a union bound, all nodes become
inactive during the next $c \cdot n \ln n$ interactions with probability at
least $1 - n^{-c/3 + 1}$. Combining, we get that all nodes become inactive
within $2c \cdot n \ln n$ interaction with probability at least $1 - 2n^{-c +
1} - n^{-c/3 + 1} \geq 1 - 3n^{-c/3 + 1}$. We can make this probability to be
at least $1 - n^{-a}$ by choosing $c = 3a + 6$.
\end{proof}

\begin{lemma}%
\label{lem:junta}
Fix any constant $a > 0$ and let $n$ be sufficiently large with respect to $a$.
With probability at least $1 - n^{-a}$ we have
\begin{equation}
     \floor{\log\log n} - 3
\leq L^*
\leq \log\log n + 4 \cdot (a + 1)
.
\end{equation}
\end{lemma}
\begin{proof}
Let $\delta \coloneqq 1/10$, $\hat{\xi}_0 = \check{\xi}_0 \coloneqq 1$, and
define for $l \in \N$
\begin{equation}
\begin{aligned}
&
          \hat{\xi}_l
\coloneqq {(1 + \delta)}^{2^l - 1} \cdot 2^{-2^{l-1}}
\\\text{and}\quad&
\check{\xi}_l
\coloneqq {(1 - \delta)}^{2^l - 1} \cdot 2^{-3 \cdot 2^{l-1} + 2}
.
\end{aligned}
\end{equation}
Note that $\hat{\xi}_l$ and $\check{\xi}_l$ are monotonically decreasing in $l$
and that for $l \in \N_0$ we have
\begin{math}
  \hat{\xi}_{l+1}
= (1 + \delta) \cdot \hat{\xi}_l^2
\end{math}
and
\begin{math}
  \check{\xi}_{l+1}
= (1 - \delta) \cdot \check{\xi}_l^2/4
\end{math}.

For the upper bound on $L^*$, apply \cref{clm:junta:aux1} and
\cref{clm:junta:aux2}.\ref{clm:junta:aux2:a}, to get that, for any $l \in \N$
with $\hat{\xi}_{l-1} \geq n^{-1/3}$ and for any constant $a > 0$,
\begin{equation}%
\label{eq:lem:junta:1}
     \Prob{B_l < \hat{\xi}_l \cdot n | B_{l-1} \leq \hat{\xi}_{l-1} \cdot n}
\geq 1 - n^{-a-1}
.
\end{equation}
(Note that, since $\hat{\xi}_0 = 1$ and $B_0 = n$, the conditioning is void for
$l = 1$.) Since $\hat{\xi}_l < n^{-1/3}$ for $l \geq \log\log n$, we can apply
\cref{eq:lem:junta:1} iteratively to see that there is an $l \leq \log\log n$
such that $\Prob{B_l < n^{2/3}} \geq 1 - l \cdot n^{-a-1}$. Together with
another application of \cref{clm:junta:aux2}.\ref{clm:junta:aux2:a}, we get an
$l \leq \log\log n + 1$ such that $\Prob{B_l < (1 + \delta) \cdot n^{1/3}} \geq
1 - l \cdot n^{-a-1}$. Combined with \cref{clm:junta:aux3} we get an $l \leq
\log\log n + 1 + 4 \cdot (a+1)$ such that $\Prob{B_l = 0} \geq 1 - l \cdot
n^{-a-1}$.

For the lower bound on $L^*$, similarly apply \cref{clm:junta:aux1} and
\cref{clm:junta:aux2}.\ref{clm:junta:aux2:b} to get that, for any $l \in \N$
with $\check{\xi}_{l-1} \geq n^{-1/3}$ and for any constant $a > 0$,
\begin{equation}%
\label{eq:lem:junta:2}
     \Prob{B_l > \check{\xi}_l \cdot n | B_{l-1} \geq \check{\xi}_{l-1} \cdot n}
\geq 1 - n^{-a-1}
\end{equation}
(As above, since $\check{\xi}_0 = 1$ and $B_0 = n$, the conditioning is void
for $l = 1$.) Since $\check{\xi}_l \geq n^{-1/3}$ for all $l \leq \log\log n -
3$, we can apply \cref{eq:lem:junta:2} iteratively to see that, for $l =
\floor{\log\log n} - 3$, $\Prob{B_l > n^{2/3}} \geq 1 - l \cdot n^{-a-1}$.

The \lcnamecref{lem:junta}'s statement follows via a union bound.
\end{proof}

\begin{lemma}%
\label{lem:junta:nodesatlowerbound}
Fix any constant $a > 0$ and let $n$ be sufficiently large with respect to $a$.
Then we have the bound
\begin{math}
     \Prob{B_{\lmax} < n^{\chosenoneminusxi/} }
\geq 1 - n^{-a}
\end{math}.
\end{lemma}
\begin{proof}
Define $\delta$ and $\hat{\xi}_l$ as in the proof of \cref{lem:junta}.
By their definition and since $\lmax = \floor{\log\log n} - 3$, we have for any
$n \in \N \setminus \set{1}$
\begin{equation}%
\label{eq:lem:junta:nodesatlowerbound:1}
\begin{aligned}
&
\hat{\xi}_{\lmax}
=
{(1 + \delta)}^{2^{\lmax} - 1} \cdot 2^{-2^{{\lmax} - 1}}
\\{}\leq{}&
{(1 + \delta)}^{2^{\log\log n - 4} - 1} \cdot 2^{-2^{\log\log n - 4 - 1}}
\\{}={}&
\frac{1}{1 + \delta} \cdot {(1 + \delta)}^{\log(n) / 16} \cdot 2^{-\log(n) / 32}
\\{}={}&
\frac{1}{1 + \delta} \cdot n^{\log(1 + \delta) / 16} \cdot n^{-1/32}
\\{}={}&
\frac{1}{1 + \delta} \cdot n^{\frac{2\log(1 + \delta) - 1}{32}}
<
n^{-\chosenxi/}
.
\end{aligned}
\end{equation}

Let $\epsilon \coloneqq 1 - \chosenxi/ = \chosenoneminusxi/$. Analogously to
the proof of \cref{lem:junta}, we have for any $l \in \N$ with $\hat{\xi}_{l-1}
\geq n^{\epsilon - 1}$ and for any constant $a > 0$
\begin{equation}%
\label{eq:lem:junta:nodesatlowerbound:2}
     \Prob{B_l < \hat{\xi}_l \cdot n | B_{l-1} \leq \hat{\xi}_{l-1} \cdot n}
\geq 1 - n^{-a-1}
.
\end{equation}
Since $\hat{\xi}_l < n^{\epsilon - 1}$ for $l \geq \lmax$ (by
\cref{eq:lem:junta:nodesatlowerbound:1} and by the monotonicity of
$\hat{\xi}_l$), we can apply \cref{eq:lem:junta:nodesatlowerbound:2}
iteratively to see that there is an $l \leq \lmax$ such that $\Prob{B_l <
n^{\epsilon}} \geq 1 - l \cdot n^{-a-1} \geq 1 - l \cdot n^{-a}$. This implies
the \lcnamecref{lem:junta:nodesatlowerbound}'s statement.
\end{proof}

\begin{proof}[Proof of~\cref{thm:junta}]
We first prove the bound on the runtime. \Cref{lem:junta:runtime} states that,
\whplong/, all nodes become inactive within $\LDAUOmicron{n \log n}$
interactions. \Cref{lem:junta} states that, \whplong/, at least one node
reaches level $\lmax$ and, thus, sets its marked and finished bits.
\Cref{lem:junta:nodesatlowerbound} states that, \whplong/, at most
$n^{\chosenoneminusxi/}$ nodes reach level $\lmax$. Finally, by
\cref{lem:epidemic} the finished bit (which spreads via a one-way epidemic) is,
\whplong/, set in all nodes after $\LDAUOmicron{n \log n}$ additional
interactions. A union bound over all these results yields the desired runtime
bound.

The number of states per node required for \ProtocolFormJuntaExt/ is
\begin{equation}
       \undervert{2}{activity bit}
\times \quad\undervert{\lmax}{level}
\quad+\quad
        \undervert{2}{marker bit}
.
\end{equation}
Note that a node's activity bit and level counter become irrelevant once its
finished bit is set (which happens at latest when reaching level $\lmax$).
Thus, when a node's finished bit is set, it leaves the $2\lmax$ states that
store the activity bit and the levels $0, 1, \dots, \lmax-1$ and enters one of
two states: one indicating that it has finished and has the marker bit not set,
and one indicating that it has finished and has the marker bit set.
\end{proof}
}
{\subsection{Phase Clock}%
\label{sec:phase_clock}

Distributed protocols often benefit from some form of synchronization.
Phase clocks~\cite{DBLP:journals/dc/AngluinAE08a} are one way to synchronize nodes in a population protocol.
The idea is to equip each node with a clock that measures time in (let's say) \emph{hours} consisting of $m \in \N$ \emph{minutes}.
These clocks do not run at a consistent speed and are not fully synchronized (a node's clock might run faster during a period in which the node is activated uncharacteristically often).
However, the clocks can be set up such that, \whplong/, each of the first $\poly(n)$ hours
\begin{enumerate}[noitemsep]
\item

    lasts $\LDAUTheta{n \log n}$ interactions for each node and

\item

    all nodes \emph{simultaneously} spend $\LDAUTheta{n \log n}$ interactions
    in each such hour.

\end{enumerate}

We adapt the phase clock implementation from~\cite{DBLP:conf/soda/GasieniecS18} to our needs.
Each node has a \emph{phase counter} $p \in \N_0$ that keeps track of the current time in minutes.
The value $m \in \N$ represents the number of minutes per hour.
Its exact value must be chosen carefully as specified by \cref{lem:phaseclock} and its proof.
The time for a node with phase counter $p$ can be expressed as $\floor{p/m}$ hours and $p \mod m$ minutes.
To limit the number of states per node, we do arithmetic on the phase counter modulo $h \cdot m$ for a parameter $h \in \N$.
We use \ProtocolPhaseClock{h} to refer to the protocol that uses the parameter $h$\footnote{%
    Technically, $m$ could also appear as a parameter in the index.
    However, for all our applications $m$ will be a constant with respect to $n$.
    Thus, we omit $m$ in the index and always assume it is chosen suitably according to \cref{lem:phaseclock}.
} (which may be a constant or grow with $n$, depending on the protocol using the phase clock).
We also allow $h = \infty$, which means that \ProtocolPhaseClock{h} uses exact phase counters that may become arbitrarily large.

We continue with a formal description of the phase clock implementation in \cref{sec:phase_clock:protocol}.
That \lcnamecref{sec:phase_clock:protocol} also states the key result (\cref{lem:phaseclock}) regarding the synchronization properties of \ProtocolPhaseClock{h}.
The protocol description is based on two simplifying assumptions.
\Cref{ssub:fixing_the_odds_and_ends} explains how to get rid of these.
To simplify the usage of the phase clock protocol in the description of other population protocols, \cref{ssub:phase_clock_interface} describes an interface and its guarantees (extracted from \cref{lem:phaseclock}) to access the phase clock, resulting in this \lcnamecref{sec:phase_clock}'s main result (\cref{thm:phase_clock}).

\subsubsection{Phase Clock Protocol \& Synchronization}%
\label{sec:phase_clock:protocol}

The state of a node is a tuple of the form $(p, b)$.
The \emph{phase counter} $p \in \N_0$ indicates the total number of minutes passed.
Initially, all nodes have $p = 0$.
The \emph{marker bit} $b \in \set{0, 1}$ indicates whether the node is marked ($b = 1$) or not ($b = 0$).
We make two simplifying assumptions for the following description:
\begin{enumerate}[noitemsep]
\item

    We assume $h = \infty$ (so we describe \ProtocolPhaseClock{\infty}).
    In particular, the phase counters are unbounded.

\item

    We assume that the number of marked nodes lies in the interval $\intcc{1, n^{1-\xi}}$ at the start of any interaction.
    Here, $\xi \in \intoc{0, 1}$ is an arbitrary constant.
    Note that the identity as well as the number of marked nodes is allowed to change arbitrarily from interaction to interaction, as long as the number of marked nodes stays within the mentioned interval.

\end{enumerate}
\Cref{ssub:fixing_the_odds_and_ends} explains how to get rid of these
assumptions.

Consider an interaction between an initiator $u$ with state $(p, b)$ and a responder in state $(p', b')$.
Protocol \ProtocolPhaseClock{\infty} causes node $u$ to update its state according to the following transition function (from~\cite{DBLP:conf/soda/GasieniecS18}):
\begin{equation}%
\label{eq:phaseclock}
\bigl[ (p, b), (p', b') \bigr] \mapsto
\begin{cases}
  (\max\set{p, p' + 1}, b)
& \text{if $b = 1$ and}
\\
  (\max\set{p, \phantomas[l]{p' + 1}{p'}}, b)
& \text{otherwise.}
\end{cases}
\end{equation}
The responder's state remains unchanged.

\paragraph{Synchronization Properties}
Remember that the $m$ denotes the number of minutes in an hour.
We define the \emph{hour} $\Round_u(t) \in \N_0$ of node $u$ with phase counter $p(t)$ after $t$ interactions as $\Round_u(t) \coloneqq \floor{p(t) / m}$.
We say a node reached hour $i \in \N_0$ if its phase counter is at least $i \cdot m$.

We now define the notion of \emph{rounds}, which represents a period of interactions during which all nodes have the same hour.
Let $\RoundStart(i)$ (\emph{start of round $i$}) denote the interaction during which the last node reaches hour $i$.
Similarly, let $\RoundEnd(i)$ (\emph{end of round $i$}) denote the interaction during which the first node reaches hour $i+1$.
If $\RoundStart(i) \leq \RoundEnd(i)$ (which is not necessarily true), then $\RoundEnd(i) - \RoundStart(i)$ equals the number of interactions $t$ for which \emph{all} nodes $u$ have, simultaneously, the same hour $\Round_u(t) = i$.
Thus, for any $i \in \N_0$ we define the \emph{length of round $i$} as $\RoundLength(i) \coloneqq \max\set{0, \RoundEnd(i) - \RoundStart(i)}$.
We also define the \emph{stretch of round $i$} as $\RoundStretch(i) \coloneqq \RoundEnd(i) - \RoundEnd(i-1)$.
In other words, the stretch of round $i$ denotes the time it takes for the first node to reach hour $i+1$ after the first node reached hour $i$.
In particular, we always have $\RoundLength(i) \leq \RoundStretch(i)$.

A key property of the above phase clock construction is captured by the following \lcnamecref{lem:phaseclock}.
It states that, by carefully choosing the phase clock parameter $m$, we can ensure that both the round length and stretch of the first $\poly(n)$ many rounds are $\LDAUTheta{n \log n}$.
It is a reformulation of~\cite[Theorem~3.1]{DBLP:conf/soda/GasieniecS18} to fit our setting and proofs.
A brief proof based on a technical lemma from~\cite{DBLP:conf/soda/GasieniecS18} is given in~\cref{app:phase-clock}.
\begin{lemma}[name=, restate=lemphaseclocks]\label{lem:phaseclock}
Let $a, c, d_1 > 0$ be constants and assume $n$ to be sufficiently large with respect to them.
There is a constant parameter $m \in \N$ (from the definition of \ProtocolPhaseClock{\infty}) and a constant $d_2 > 0$ such that, with probability at least $1 - n^{-a}$, for all $i \in \set{0, 1, \dots, n^{c}}$
\begin{enumerate}
\item
    \begin{math}
         \RoundLength(i)
    \geq d_1 \cdot n \log n
    .
    \end{math}

\item
    \begin{math}
         \RoundStretch(i)
    \leq d_2 \cdot n \log n
    .
    \end{math}
\end{enumerate}
\end{lemma}
Note that in the above \lcnamecref{lem:phaseclock}, the constant parameter $m$ depends on the involved constants $a$, $c,$ and $d_1$.
In particular, it increases with the exponent $a$ of the desired error probability.

\subsubsection{Fixing the Odds and Ends}%
\label{ssub:fixing_the_odds_and_ends}

We briefly explain how the simplifying assumptions we made for the protocol description can be avoided.

\paragraph{Computing a Junta On the Fly}
Our protocol description in \cref{sec:phase_clock:protocol} assumes that the number of marked nodes lies in the interval $\intcc{1, n^{1 - \xi}}$ at the start of any interaction, where $\xi \in \intoc{0, 1}$ is an arbitrary constant.
Instead of assuming a priori the existence of such a junta in each round, we can use protocol \ProtocolFormJuntaExt/ from \cref{sec:junta_creation} to generate such a junta (with $\xi = \chosenxi/$) \whplong/ within $\LDAUOmicron{n \log n}$ interactions using $2 \cdot (\floor{\log\log n} - 2)$ states per node (see \cref{thm:junta}).

The state of a node is a tuple $(l, a, b, f, p)$.
The (sub-) tuple $(l, a, b, f)$ is used as the state for the junta protocol and consists of the level $l \in \set{0, 1, \dots, \floor{\log\log n} - 3}$, the activity bit $a \in \set{0, 1}$, the marker bit $b \in \set{0, 1}$, and the finished bit $f \in \set{0, 1}$.
The (sub-) tuple $(p, b)$ is used for the phase clock protocol and consists of the phase counter $p \in \N_0$ and the marker bit $b \in \set{0, 1}$.
Note that the marker bit $b$ is used by both protocols.
All nodes start in state $(0, 1, 0, 0, 0)$ (with only the activity bit set) and execute protocol \ProtocolFormJuntaExt/ on the (sub-) tuple $(l, a, b, f)$.
Once the finished bit $f$ of a node is set it starts to execute the phase clock protocol from \cref{sec:phase_clock:protocol} on the (sub-) tuple $(p, b)$.
We make two simple observations:
\begin{enumerate}[noitemsep]
\item

    \ProtocolPhaseClock{\infty} starts only when (if) the first node in \ProtocolFormJuntaExt/ becomes marked (and, thus, finished).
    By \cref{thm:junta}, this happens \whplong/ within $\LDAUOmicron{n \log n}$ interactions.
    Additionally, since the finished bit spreads via a one-way epidemic, \whplong/ \emph{all} nodes start to execute \ProtocolPhaseClock{\infty} in $\LDAUOmicron{n \log n}$ interactions (by \cref{lem:epidemic}).

\item

    When \ProtocolPhaseClock{h} starts, it does so with a junta of size exactly $1$.
    During its execution, the junta might grow (due to more nodes becoming marked in \ProtocolFormJuntaExt/).
    However, by \cref{thm:junta}, \whplong/ the junta does not grow beyond size $n^{\chosenoneminusxi/}$.

\end{enumerate}

It follows that \cref{lem:phaseclock} holds also if the junta is computed on the fly, with the only difference being that it takes $\LDAUOmicron{n \log n}$ interactions before the first node starts to increase its phase counter.
This yields the following observation.
\begin{observation}%
\label{obs:junta_on_the_fly}
We can change \ProtocolPhaseClock{\infty} such that, \whplong/, it computes a non-empty junta (marked nodes) of size at most $n^{\chosenoneminusxi/}$ on the fly and \cref{lem:phaseclock} still holds.
\end{observation}

\paragraph{Unbounded State Space}
The population protocol as described in \cref{sec:phase_clock:protocol} requires an unbounded number of states, since a node's phase counter $p$ is unbounded.
We can avoid this by performing any arithmetic on the phase counter modulo $h \cdot m$.
Here, $h \in \N$ is a parameter that controls how many hours nodes can count exactly (a node reaching hour $h$ cannot tell whether it has hour $h$ or hour $0$).

Note that \cref{lem:phaseclock} implies that during the first $\poly(n)$ many rounds all nodes are, \whplong/, in neighboring hours (otherwise, if there was a time where one node is in hour $i$ and another node is in hour $i+2$, those nodes could never be simultaneously in hour $i+1$).
Thus, $h = 3$ is already enough to allow a node, \whplong/, to distinguish whether its interaction partner is in an earlier, in the same, or in a later hour.
We get the following observation.
\begin{observation}%
\label{obs:bounding_the_states}
Assume $h \geq 3$.
Define \ProtocolPhaseClock{h} analogously to \ProtocolPhaseClock{\infty} (see \cref{eq:phaseclock} but with all arithmetic on the phase counter $p$ done modulo $h \cdot m$.
\Whplong/, all nodes can correctly determine the maximum in the transition function of \ProtocolPhaseClock{h} (\cref{eq:phaseclock}) during the first $n^c$ rounds, where $c$ is the constant from \cref{lem:phaseclock}.
In particular, \cref{lem:phaseclock} holds also for \ProtocolPhaseClock{h}.
\end{observation}

\subsubsection{Phase Clock Interface}%
\label{ssub:phase_clock_interface}

To simplify the usage of the phase clock in our Majority protocols, we provide an interface to \ProtocolPhaseClock{h}, together with the guarantees implied by \cref{lem:phaseclock}.
The parameter $h \in \N \cup \set{\infty}$ is assumed to be at least $3$.
We group the guarantees of the different interface functions in three categories:
\begin{enumerate}[(A), wide=0pt, itemsep=0.618em]
\item\label[guarantee]{it:phase_clock_interface:A}
    The following function calls to \ProtocolPhaseClock{h} are guaranteed to work as described with probability $1$:
    \begin{itemize}[noitemsep, wide=0em]
    \item \PCcall{h}{u, v}:
        Update the state of $u$ according to \cref{eq:phaseclock}.

    \item \PCmarked{u}:
        Return \bTrue/ iff $u$'s marker bit $b$ is set (meaning $u$ is a junta node).

    \item \PCfinishedJunta{u}:
        Return \bTrue/ iff $u$'s finished bit $f$ is set.

    \item \PCoverflowed{u}:
        Return \bTrue/ iff, in the past, the phase counter of $u$ decreased at least once in absolute value (due to the modulo $h \cdot m$ arithmetic).

    \item \PCnewHour{u}:
        Return \bTrue/ iff $u$ reached a new hour the last time it updated the phase counter.

    \item \PCskippedHour{u}:
        Return \bTrue/ iff there was an interaction during which the hour of node $u$ increased by at least $2$ (this happens if the clocks are not properly synchronized).
    \end{itemize}

\item\label[guarantee]{it:phase_clock_interface:B}
    The following function calls to \ProtocolPhaseClock{h} are guaranteed to work as described for $n^c$ many rounds with probability $1 - n^{-a}$ for any constants $a, c > 0$ (assuming $m$ was chosen suitably and $n$ is sufficiently large):
    \begin{itemize}[noitemsep, wide]
    \item \PCdifferentHour{u, v}:
        Return \bTrue/ iff $u$ is currently in a different hour as $v$.

    \item \PCsameHour{u, v}:
        Return \bTrue/ iff $u$ is currently in the same hour as $v$.

    \item \PCsmallerHour{u, v}:
        Return \bTrue/ iff $u$ is currently in a smaller hour than $v$.

    \item \PClargerHour{u, v}:
        Return \bTrue/ iff $u$ is currently in a larger hour than $v$.
    \end{itemize}

\item\label[guarantee]{it:phase_clock_interface:C}
    Moreover, until the first node reaches hour $h$ (i.e., for the first $\RoundEnd(h-1)$ many interactions), all function calls work as described with probability $1$.
\end{enumerate}

Protocols using the phase clock should be aware that, with negligible probability, the phase clock might not run at all (no nodes were marked) or might run too fast (if $n^{1 - \LDAUomicron{1}}$ nodes were marked).

We gather the above guarantees in \cref{thm:phase_clock}, the main result of this \lcnamecref{sec:phase_clock}.
In the following, remember that $\lmax = \floor{\log\log n} - 3$ is the maximum junta level.
\begin{theorem}%
\label{thm:phase_clock}
Let $a, c > 0$ be constants and assume $n$ to be sufficiently large with respect to them.
Consider a parameter $h \in \set{3, 4, \dots} \cup \set{\infty}$.
\ProtocolPhaseClock{h} supports the interface specified above with \cref{it:phase_clock_interface:A, it:phase_clock_interface:B, it:phase_clock_interface:C} and uses $\LDAUTheta{h + \log\log n}$ states per node.
A node whose phase clock is running (finished bit from junta creation is set) is in one of $\LDAUTheta{h}$ many states.
\end{theorem}
\begin{proof}[Proof of \cref{thm:phase_clock}]
\PCcall{h}{\cdot}, \PCmarked{\cdot}, as well as \PCfinishedJunta{\cdot} are simple state updates and lookups.
As such, they are correct by definition.
The function calls \PCoverflowed{\cdot}, \PCnewHour{\cdot}, and \PCskippedHour{\cdot} can be implemented by providing a bit for each of them, which is set to either \bTrue/ or \bFalse/ according to the respective function description (note that the corresponding conditions can be easily checked locally by a node).
This implies \cref{it:phase_clock_interface:A}.

The statements from \cref{it:phase_clock_interface:B} (which cover the function calls \PCdifferentHour{\cdot}, \PCsameHour{\cdot}, \PCsmallerHour{\cdot}, and \PClargerHour{\cdot}) are a consequence of the choice $h \geq 3$ and \cref{lem:phaseclock, obs:bounding_the_states}.
These ensure that, \whplong/, for $\poly(n)$ rounds, the hours of any pair of nodes differ by at most one.

Before the first node reaches hour $h$, nodes store their exact phase counter and, thus, know their exact hour.
This implies \cref{it:phase_clock_interface:C}.

We now bound the number of states each node requires.
By \cref{thm:junta}, the on the fly creation of the junta requires $2 \cdot (\lmax + 1)$ states.
Note that, as described in \cref{ssub:fixing_the_odds_and_ends}, the values of a node's phase clock state (marker bit, phase counter, bit for \PCoverflowed{\cdot}, bit for \PCnewHour{\cdot}, bit for \PCskippedHour{\cdot}) only become relevant once its finished bit from the junta creation is set.
At that moment, nodes can forget the level from the junta calculation and only need to store whether they are finished and marked or finished and unmarked.
Thus, for each of the two value of the marker bit when the node is finished, the maximum number of states that can occur is bounded by $h \cdot m \times 2^3$.
So in total, the number of states per node is
\begin{equation}
       \undervert{2 \cdot \lmax}{\shortstack{junta\\calculation}}
+      \undervert{2}{marked?}
\times \undervert{h \cdot m}{phase counter}
\times \undervert{2^3.}{\renewcommand{\arraystretch}{0.5}\begin{tabular}{l}\PCoverflowed{\cdot}\\\PCnewHour{\cdot}\\\PCskippedHour{\cdot}\end{tabular}}
\end{equation}
Since we have $\lmax = \LDAUTheta{\log\log n}$ and $m = \LDAUTheta{1}$, this is
$\LDAUTheta{h + \log\log n}$.
\end{proof}
}
{\section{Simple Majority}%
\label{sec:majority_clocks}

In this section we present and analyze our first majority protocol, \ProtocolSimpleMaj{s, h}, which works correctly \whplong/.
It is parameterized by two integer values $s$ and $h$ (the latter value is used for the phase clocks).
As many majority protocols, it is based on a variant of the \emph{cancellation} and \emph{doubling} approach (see, e.g.,~\cite{DBLP:journals/dc/AngluinAE08a}).
Here, the  general idea is that nodes first perform cancellation (opposite opinions cancel each other out) for $\LDAUTheta{n \log n}$ consecutive interactions and then, for another $\LDAUTheta{n \log n}$ consecutive interactions, each node with an opinion finds a node whose opinion was canceled and copies its opinion onto that node.
Cancellation boost the ratio between the support of majority and minority opinions, while duplication eliminates non-opinionated nodes created during cancellation.

Our protocol uses \emph{cancellation} as described above.
However, nodes do  not simply create a single copy of their opinion but $s \geq 2$ copies (\emph{load explosion}).
These copies are distributed among the nodes via a load balancing mechanism.
This approach allows us to reduce the number of required phases.
We will prove the following theorem:
\begin{theorem}[name=, restate=thmmajorityclocks]\label{thm:majority_clocks}
Let $s \in \N \setminus \set{1}$ and $h \in \N \setminus \set{1, 2}$.
Consider the majority problem for $n$ nodes with initial absolute bias $\BiasAbsolute \in \N$.
\Whplong/, protocol \ProtocolSimpleMaj{s, h} correctly identifies the majority for all interactions $t = \LDAUOmega{n \log n \cdot \log_s(n / \BiasAbsolute)}$.
It uses $\LDAUTheta{h s + \log\log n}$ states per node.
\end{theorem}
According to \Cref{thm:majority_clocks} there is no  benefit by choosing $h > 3$.
However, our stable protocol presented in \cref{sec:majority_clocks_exact} does rely on a larger value of $h$.

We now describe the protocol's state space and its transition function (see also \cref{alg:clocks}).
Afterward, we give the proof of \cref{thm:majority_clocks}.

{\begin{algorithm}
\SetKw{KwLet}{let}
\SetKwProg{Algorithm}{}{}{}

\SetKw{KwAnd}{and}
\SetKw{KwLet}{let}
\SetKw{KwNot}{not}
\SetKw{KwOr}{or}

\linesOff
\Algorithm{$\ProtocolSimpleMaj{s, h}(u,v)$}{
    \linesOn

    \PCcall{h}{u, v}\tcc*[r]{synchronization}
    \BlankLine\BlankLine

    \If(\tcc*[f]{load explosion }){\PCnewHour{u}}{
        $\Load{u} \gets \Load{u} \cdot s$\;
    }
    \BlankLine\BlankLine

    \If(\tcc*[f]{load balancing }){\PCsameHour{u, v}}{
        \label{ln:simplalg:samehour}
        $(\Load{u}, \Load v) \gets \left(\ceil*{\frac{\Load{u} + \Load v}{2}}, \floor*{\frac{\Load{u} + \Load v}{2}}\right)$\;
    }
}

\caption{%
    Pseudocode for transition function of \ProtocolSimpleMaj{s, h} (initiator $u$ and responder $v$).
}\label{alg:clocks}
\end{algorithm}
}

\paragraph{State Space}
The state of a node $u$ consists of the states required for the \ProtocolPhaseClock{h} protocol (which subsumes the states of \ProtocolFormJuntaExt/, cf.~\cref{sec:phase_clock}) and a \emph{load value} $\Load{u}$.
The load value $\Load{u}$ represents $u$'s current opinion (sign) and its \enquote{magnitude} (absolute value).
It is initialized with either $+1$ or $-1$, depending on $u$'s initial opinion.
The output function maps the state of a node to the sign of its load value.
Thus, the majority guess of a node $u$ is equal to $\sign(\Load{u})$.\footnote{%
    The value $\sign(\Load{u}) = 0$ (i.e., $\Load{u} = 0$) can be interpreted as an \enquote{undecided} opinion.
    In the proof of \cref{thm:majority_clocks} we show that, \whplong/, all nodes eventually agree on a non-zero sign value.
}

For most of the analysis, we assume unbounded, integral load values.
In the proof of \cref{thm:majority_clocks}, we will see that, \whplong/, load values will be integers not exceeding $3s$ unless all nodes have already the same positive or negative sign.
This allows us to \emph{cap the absolute load values} at $3s$ (i.e., setting a node $u$'s load via the assignment $\Load{u} \gets x$ to a value $x$ with $\abs{x} \geq 3s$ instead sets $\Load{u} \gets \sign(x) \cdot 3s$) while still maintaining the high probability guarantee from \cref{thm:majority_clocks}.

\paragraph{Transition Function}
Consider an interaction between two nodes $u$ (initiator) and $v$ (responder).
The nodes' actions can be divided into three parts: \emph{synchronization}, \emph{load explosion}, and \emph{load balancing}.
During the synchronization, the \ProtocolPhaseClock{h} protocol is triggered with initiator $u$ and responder $v$ to update the states of $u$'s phase clock.
During the load explosion, $u$ uses the $\PCnewHour{\cdot}$ method to check whether this is its first interaction in its current hour.
If yes, it multiplies its load by a factor of $s$.
During the load balancing, the nodes use the phase clock's $\PCsameHour{\cdot}$ method to check whether they are in the same hour and, if so, perform a simple load balancing step by balancing their respective loads as evenly as possible.

The following \lcnamecref{obs:stable:benefit_h} follows from the fact that all phase clock function calls work correctly with probability $1$ until the first node reaches hour $h$ (\cref{it:phase_clock_interface:C} in \cref{ssub:phase_clock_interface}).
In particular, since nodes $u$ and $v$ balance their loads only if $\PCsameHour{u, v}$ returns \bTrue/ (\cref{ln:simplalg:samehour} in \cref{alg:clocks}), we get:
\begin{observation}%
\label{obs:stable:benefit_h}
Whenever two nodes $u$ and $v$ balance their loads in \ProtocolSimpleMaj{s,h} before some node reached hour $h$, both $u$ and $v$ are guaranteed to be in the same hour.
\end{observation}
This observation will be of importance for our stable majority protocol presented in \cref{sec:majority_clocks_exact} (which is based on \ProtocolSimpleMaj{s, h}).

\paragraph{Total \& Scaled Total Load}
Let $\Load{u}(t)$ denote the load of node $u$ after $t$ interactions and $\LExpl{u}(t)$ the number of load explosions node $u$ experienced after $t$ interactions.
Define the \emph{total load} $\TLoad{t}$ and the \emph{scaled total load} $\STLoad{t}$ after $t$ interactions as
\begin{dmath*}%
          \TLoad{t}
\coloneqq \sum_{u \in \intcc{n}} \Load{u}(t)
\quad\text{and}\quad
          \STLoad{t}
\coloneqq \sum_{u \in \intcc{n}} \frac{\Load{u}(t)}{s^{\LExpl{u}(t)}}
\end{dmath*}.
Observe that $\STLoad{0} = \TLoad{0}$ is the total initial load.
Thus, $\sign(\STLoad{0}) = \sign(\TLoad{0})$ reflects the initial majority and $\abs{\STLoad{0}} = \abs{\TLoad{0}}$ equals the initial absolute bias $\BiasAbsolute$.

The following \lcnamecref{lem:correct_opinion} will be useful to show that, if \ProtocolSimpleMaj{s, h} works for $\LDAUOmicron{\log n}$ rounds as expected (the phase clock runs, is properly synchronized, and the loads balance out), all nodes forever agree on the correct initial majority.
\begin{lemma}%
\label{lem:correct_opinion}
Let $t \in \N_0$ and assume that whenever two nodes $u$ and $v$ balance their loads in an interaction $t' \leq t$, $\LExpl{u}(t') = \LExpl{v}(t')$.
Then $\STLoad{t} = \STLoad{0}$.
If, additionally, for all nodes $u$ and $v$ we have $\sign(\Load{u}(t)) = \sign(\Load{v}(t))$, then all nodes forever agree on the correct initial majority opinion after interaction $t$.
\end{lemma}
\begin{proof}
The invariant for $\STLoad{t}$ follows via a simple induction over $t$.
For the second part, assume all nodes' load values have the same sign after $t$ interactions.
Note that no load balancing action can change this, afterward.
Moreover, the total scaled load $\STLoad{t}$ also has the same sign as each single node.
So every node's sign forever equals $\sign(\STLoad{t})$.
Since the \lcnamecref{lem:correct_opinion}'s first part states $\sign(\STLoad{t}) = \sign(\STLoad{0})$ (the initial majority opinion), this implies that each node's sign forever equals the correct initial majority opinion after interaction $t$.
\end{proof}

We are now ready to prove this \lcnamecref{sec:majority_clocks}'s main result.
\begin{proof}[Proof of \cref{thm:majority_clocks}]
For $i \in \N_0$ let $T_i$ denote the last interaction of round $i$ (with $T_i = \infty$ if $\RoundLength(i) = 0$).
Define $i^* \coloneqq \ceil{\log_s(2n / \BiasAbsolute)}$.
Applying \cref{lem:loadbalancing, lem:phaseclock}, with $d_1$ from \cref{lem:phaseclock} equal to the constant $c$ from \cref{lem:loadbalancing}, and using a union bound over the first $i^* + 1 = \LDAUOmicron{\log n}$ rounds yields that, \whplong/, the following properties hold:
\begin{enumerate}[(1)]
\item\label[prop]{prf:majority_clock:prop1} For all $i \in \set{0, 1, \dots, i^*}$, we have $\RoundLength(i) = \LDAUOmega{n \log n}$ and $\RoundStretch(i) = \LDAUOmicron{n \log n}$ (\cref{lem:phaseclock}).
\item\label[prop]{prf:majority_clock:prop2} For all $i \in \set{0, 1, \dots, i^*}$, the loads have discrepancy at most $2$ after interaction $T_i$ (\cref{lem:loadbalancing}).
\end{enumerate}
Note that \cref{prf:majority_clock:prop1} implies that no node skips any hour $i \in \set{0, 1, \dots, i^*}$:
If there were such a node, it had hour $< i$ and met a node in hour $> i$ when it skipped hour $i$.
But then, by definition of a round's length, we have $\RoundLength(i) = 0$.
This contradicts \cref{prf:majority_clock:prop1}.

We now condition on the high probability event that the above properties hold and consider the first $T_{i^*}$ interactions.
Because nodes are properly synchronized, the calls to $\PCsameHour{\cdot}$ (\cref{ln:simplalg:samehour}) correctly indicate whether two nodes are in the same hour or not.
Also, since no node skipped an hour, any node in hour $i$ experienced exactly $i$ load explosions.
This implies that, whenever two nodes balance their loads during the first $T_{i^*}$ interactions, they experienced the same number of load explosions.
\Cref{lem:correct_opinion} gives $\STLoad{T_{i^*}} = \STLoad{0}$, and the scaled total load definition gives
\begin{math}
     \abs{\STLoad{T_{i^*}}}
=    \abs{\TLoad{T_{i^*}}} / s^{i^*}
\leq \BiasAbsolute \cdot \abs{\TLoad{T_{i^*}}} / (2n)
\end{math}.
By using
\begin{math}
  \abs{\STLoad{T_{i^*}}}
= \abs{\STLoad{0}}
= \alpha
\end{math}
this yields
\begin{math}
     \abs{\TLoad{T_{i^*}}}
\geq 2n
\end{math}.

Note that if $\abs{\TLoad{T_{i^*}}}\ge 2n$, the absolute value of the average load is at least $2$.
Hence, either all nodes have load exactly $2$ (or $-2$), or there is at least one node with load $\ge 3$ ($\le -3$).
In the later case it follows from  \cref{prf:majority_clock:prop2} that all other nodes have load at least $1$ (at most $-1$).
In both cases, all loads have the same sign after interaction $T_{i^*}$.
Thus, using again \cref{lem:correct_opinion}, all nodes forever agree on the correct initial majority opinion after interaction $T_{i^*}$.
The runtime bound follows since, by \cref{prf:majority_clock:prop1}, the first $i^* + 1 = \LDAUOmicron{\log_s(n / \alpha)}$ rounds have stretch $\LDAUOmicron{n \log n}$.

To bound the number of states, observe that – conditioned on the high probability event that the above properties hold – no absolute load value exceeds $2s$ unless all nodes' loads have the same sign.
Indeed, if not all loads have the same sign at the end of a round, the discrepancy bound (\cref{prf:majority_clock:prop2}) implies that all loads have absolute load at most $2$.
After the load explosion in the following round the load will be at most $2s$.
This allows us to cap the absolute load values at $3s$ as described at the beginning of this \lcnamecref{sec:majority_clocks} and, \whplong/, the protocol outcome will not change.\footnote{
    For \ProtocolSimpleMaj{s, h}, we could also cap at $2s$.
    The cap at $3s$ is used in our stable majority protocol in \cref{sec:majority_clocks}.
    Note that, if the load balancing works as expected (discrepancy $\leq 2$), any node with load $3s$ can be sure that \emph{all} loads have the same, non-zero sign.
}
These load values are combined with the states from \ProtocolPhaseClock{h}.
By \cref{thm:phase_clock}, that protocol requires in total $\LDAUTheta{h + \log\log n}$ states per node, but only $\LDAUTheta{h}$ states per node once the node has finished the junta election process.
From that time on, each node needs to store the current phase of the \ProtocolPhaseClock{h} protocol and the current load value.
Thus, \ProtocolSimpleMaj{s, h} requires $\LDAUTheta{h s + \log\log n}$ states per node.
\end{proof}
}
{\section{Stable Majority}%
\label{sec:majority_clocks_exact}

In this \lcnamecref{sec:majority_clocks_exact}, we present and analyze the protocol \ProtocolStableMaj{s}, a hybrid majority protocol which stabilizes efficiently.
We prove the following \lcnamecref{thm:majority_clocks_exact}:
\begin{theorem}%
\label{thm:majority_clocks_exact}
Let $s \in \set{2, 3, \dots, n}$.
Consider the majority problem for $n$ nodes with initial absolute bias $\BiasAbsolute \in \N$.
Protocol \ProtocolStableMaj{s} is exact and stabilizes \whplong/ and in expectation in $\ldauOmicron{n \log n \cdot \log_s(n / \BiasAbsolute)}$ interactions.
It uses $\LDAUTheta{s \cdot \log_s n}$ states per node.
\end{theorem}

We now describe the protocol's state space and its transition function (see also \cref{alg:exact}).
Afterward, we give the proof of \cref{thm:majority_clocks_exact}.
{\begin{algorithm}
\SetKw{KwLet}{let}
\SetKwProg{Algorithm}{}{}{}

\SetKw{KwLet}{let}
\SetKw{KwAnd}{and}
\SetKw{KwOr}{or}
\SetKwFunction{Exit}{exit}
\SetKwFunction{Balance}{balance}
\SetKwFunction{Multiply}{multiply}
\SetKwFunction{Wait}{wait}

\linesOff
\Algorithm{\ProtocolStableMaj{s}$(u,v)$}{
    \linesOn

    $\ProtocolBackupMaj/(u, v)$\;
    \label{ln:stablealg:backup}
    \BlankLine\BlankLine

    \lIf{$\Finished{v}$}{
        $\Finished{u} \gets \bTrue/$
    }
    \label{ln:stablealg:onewayep:finished}
    \lIf{$\phantomas[l]{\Finished{v}}{\Error{v}}$}{
        $\phantomas[l]{\Finished{u}}{\Error{u}} \gets \bTrue/$
    }
    \label{ln:stablealg:onewayep:error}
    \BlankLine\BlankLine

    \If{$\neg\Error{u} \land\; \neg\Finished{u}$}{
        \label{ln:stablealg:callsimple:start}
        $\ProtocolSimpleMaj{s, h}(u, v)$\;
    }
    \label{ln:stablealg:callsimple:stop}
    \BlankLine\BlankLine

    \If{$\PCoverflowed{u} \lor \abs{\Load{u}} \geq 3s$}{
        \label{ln:stablealg:finished:start}
        $\Finished{u} \gets \bTrue/$
    }
    \label{ln:stablealg:finished:stop}
    \BlankLine\BlankLine

    \If{$\bigl(\Finished{u} \land \Finished{v} \land \sign(\Load{u}) \neq \sign(\Load{v})\bigr)$ ${}\lor \PCskippedHour{u}$}{
        \label{ln:stablealg:error:start}
        $\Error{u} \gets \bTrue/$\;
    }
    \label{ln:stablealg:error:stop}
}
\caption{%
    Pseudocode for transition function of \ProtocolStableMaj{s} (initiator $u$ and responder $v$).
    It calls \ProtocolSimpleMaj{s, h} with $h \coloneqq \ceil{\log_s(4n)} + 2$.
}\label{alg:exact}
\end{algorithm}
}

Each node $u$ executes a slow but exact protocol \ProtocolBackupMaj/\footnote{%
    We use the $4$-state protocol from~\cite{DBLP:conf/icalp/MertziosNRS14} for this, which stabilizes in $\LDAUOmicron{n^2 \log n}$ interactions in expectation, implying a finite stabilization time and, thus, exactness.
} (\cref{ln:stablealg:backup}) as well as up to $h$ rounds of our fast but possibly incorrect \ProtocolSimpleMaj{s, h} (\crefrange{ln:stablealg:callsimple:start}{ln:stablealg:callsimple:stop}), with $h \coloneqq \ceil{\log_s(4n)} + 2$.
As output, we use the output of the backup protocol if the phase clock is not yet running ($u$'s phase counter is zero and $\PCoverflowed{u} = \bFalse/$) or if $u$ thinks that protocol \ProtocolSimpleMaj{s, h} failed (an error bit is set).
Otherwise, we use the output of \ProtocolSimpleMaj{s, h}.

Node $u$ stops \ProtocolSimpleMaj{s, h} via a \emph{finished bit} \Finished{u} and checks whether \ProtocolSimpleMaj{s, h} failed via an \emph{error bit} \Error{u}.
Both bits are initially \bFalse/ and are spread via a one-way epidemic  (\crefrange{ln:stablealg:onewayep:finished}{ln:stablealg:onewayep:error}).
\ProtocolSimpleMaj{s, h} is executed only while both bits are \bFalse/ (\crefrange{ln:stablealg:callsimple:start}{ln:stablealg:callsimple:stop}).

The (first) finished bit is set for one of two reasons (\crefrange{ln:stablealg:finished:start}{ln:stablealg:finished:stop}):
\begin{enumerate*}[(i)]
\item $u$ reached hour $h$ (i.e., its phase counter overflowed).
    This marks the end of the first $h$ rounds.
    Stopping at this point ensures that any load balancing operation happens between two nodes in the same hour (\cref{obs:stable:benefit_h}).
\item $u$ has absolute load at least $3s$ after its first\footnote{%
        Note that the absolute load of a node $u$ can only increase to $\geq 3s$ because of a load explosion.
        So when the condition $\abs{\Load{u}} \geq 3s$ holds for the first time, that node just went through a load explosion and, thus, just entered a new hour.
    } interaction in an hour $i$.
    Then it had absolute load at least $3$ at the end of round $i-1$.
    If \ProtocolSimpleMaj{s, h} managed to balance the loads during round $i-1$, the load of any other node differs by at most $2$.
    Thus, all nodes have the same sign, which we will show to be correct if no node sets its error bit.
\end{enumerate*}

The (first) error bit is also set for one of two reasons (\crefrange{ln:stablealg:error:start}{ln:stablealg:error:stop}):
\begin{enumerate*}[(i)]
\item Two finished nodes whose loads have different signs interact with each other, in which case \ProtocolSimpleMaj{s, h} obviously failed.
\item A node skipped an hour.
    Then it is no longer true that a node in hour $i$ experienced exactly $i$ load explosions.
    This might cause \ProtocolSimpleMaj{s, h} to fail, since two nodes that experienced a different number of load explosions might balance their loads.
\end{enumerate*}

Since the backup protocol is exact, our protocol is exact if the error bit is set.
A major part of the analysis is to show that it is also exact if none of the error bits is set.
Moreover, we have to show that, with high probability, no error bit is set and the protocol stabilizes fast.
\begin{proof}[Proof of \cref{thm:majority_clocks_exact}]
Let us first bound the number of states per node.
By \cref{thm:majority_clocks}, \ProtocolSimpleMaj{s, h} requires $\LDAUTheta{h s + \log\log n} = \LDAUTheta{s \cdot \log_s n}$ states.
This is increased by a constant factor from the $4$ states for \ProtocolBackupMaj/ and the $4$ combinations of the bits \Finished{u} and \Error{u}, yielding the desired bound.

Next, we prove that \ProtocolStableMaj{s} is exact.
That is, if \TimeST/ denotes the stabilization time of protocol \ProtocolStableMaj{s}, we show that $\TimeST/ < \infty$ with probability $1$.
We distinguish three cases:
\begin{enumerate}[(i), wide=1em]
\item\label[case]{it:majority_clocks_exact:exactcases:i} \emph{The phase clock does not start:}
    That is, in \ProtocolFormJuntaExt/ all nodes set their activity bit to $0$ before reaching level $\lmax$.
    No node is marked, such that the phase counters cannot increase and $\PCoverflowed{u}$ always returns $\bFalse/$.
    Then all nodes forever use the output of the backup protocol, which has finite stabilization time.
    Thus, $\TimeST/ < \infty$ in this case.

\item\label[case]{it:majority_clocks_exact:exactcases:ii} \emph{The phase clock starts and some node sets its error bit:}
    The error bit is spread via a one-way epidemic (\cref{ln:stablealg:onewayep:error}).
    Thus, with probability $1$ eventually all nodes set their error bit.
    From then on, they use the output of the backup protocol, yielding again $\TimeST/ < \infty$.

\item\label[case]{it:majority_clocks_exact:exactcases:iii} \emph{The phase clock starts and no node ever sets its error bit:}
    If the phase clock runs, \ProtocolFormJuntaExt/ marks at least one node and, eventually, all nodes $u$ set their finished bit \Finished{u}:
    Indeed, nodes with an unset finished bit execute the phase clock (via \ProtocolSimpleMaj{s, h}), such that they have a non-zero probability to increase their phase counter (since there is a marked node, see \cref{sec:phase_clock:protocol}).
    Thus, eventually the phase counter overflows and the finished bit is set (\crefrange{ln:stablealg:finished:start}{ln:stablealg:finished:stop}).

    Let $T < \infty$ denote the interaction after which all finished bits are set.
    Since no error bit is ever set, no node ever skips an hour (\crefrange{ln:stablealg:error:start}{ln:stablealg:error:stop}).
    Thus, any two nodes that balance their loads are not only in the same hour $i < h$ (as checked by \ProtocolSimpleMaj{s, h}) but also experienced both exactly $i$ load explosions.
    Moreover, after interaction $T$ the loads no longer change and all nodes have the same load signs (otherwise, eventually two finished nodes of different sign meet and an error bit is set, contradicting the case assumption).
    Thus, by \cref{lem:correct_opinion}, all nodes forever agree on the correct initial majority after interaction $T$, such that $\TimeST/ \leq T < \infty$.
\end{enumerate}

It remains to prove the \lcnamecref{thm:majority_clocks_exact}'s runtime bounds. We first show that $\TimeST/ = \ldauOmicron{n \log n \cdot \log_s(n / \BiasAbsolute)}$ \whplong/.

To this end, let $H^* \in \N$ denote the maximal hour ever reached by any node and for $i \in \set{0, 1, \dots, H^*-1}$ let $T_i$ be the last interaction of round $i$.
Define $T^*$ as the first interaction during which some node sets its finished or error bit.
By \crefrange{ln:stablealg:finished:start}{ln:stablealg:error:stop}, the \emph{first} finish or error bit is set because of three possible reasons:
\begin{enumerate}[(i)]
\item\label[reas]{it:prf:stablmaj:reason:1} a node had load at least $3s$ after its first interaction in an hour (finished bit),
\item\label[reas]{it:prf:stablmaj:reason:2} a node's phase counter overflowed (finished bit), or
\item\label[reas]{it:prf:stablmaj:reason:3} a node skipped an hour (error bit).
\end{enumerate}
In a similar way to the proof of \cref{thm:majority_clocks}, we first show that, \whplong/, \cref{it:prf:stablmaj:reason:1} applies and that all nodes agree on the correct initial majority after $T^*$ interactions without setting the error bit.
At that moment, we might not yet have stabilized, since there's still a non-zero probability for a node to set the error bit because of \cref{it:prf:stablmaj:reason:3}.
But \whplong/ that won't happen before all nodes set their finished bit by the one-way epidemic (\cref{ln:stablealg:onewayep:finished}), after which the error bit cannot be set anymore.
We formalize this idea below.

Note that the finished bit is set when a node reaches hour $h$ (its phase counter overflows), so $H^* \leq h$.
As in \cref{thm:majority_clocks}'s proof, we apply \cref{lem:phaseclock, lem:loadbalancing} via a union bound to the first $H^* \leq h = \LDAUOmicron{\log n}$ rounds to get, \whplong/, the following properties:
\begin{enumerate}[(1)]
\item\label[prop]{it:prf:stablmaj:fast:0} $T^* < \infty$ (the phase clock runs and some node sets its error bit or, eventually, its phase counter overflows).
\item\label[prop]{it:prf:stablmaj:fast:1} For all $i \in \set{0, 1, \dots, H^*-1}$, we have $\RoundLength(i) = \LDAUOmega{n \log n}$ and $\RoundStretch(i) = \LDAUOmicron{n \log n}$ (\cref{lem:phaseclock}).
\item\label[prop]{it:prf:stablmaj:fast:2} For all $i \in \set{0, 1, \dots, H^*-1}$, the loads have discrepancy at most $2$ after interaction $T_i$ (\cref{lem:loadbalancing}).
\end{enumerate}
As in the proof of \cref{thm:majority_clocks}, \cref{it:prf:stablmaj:fast:1} implies that no node ever skips an hour.

In the remainder we condition on the high probability event that the above \lcnamecrefs{it:prf:stablmaj:fast:0} hold.
Since no node ever skips an hour, whenever two nodes in hour $i \in \set{0, 1, \dots, H^* - 1}$ balance their loads, both of them experienced exactly $i$ load explosions.
Thus, \cref{lem:correct_opinion} gives $\STLoad{T_i} = \STLoad{0}$ for all $i \in \set{0, 1, \dots, H^* - 1}$.
With this, we can show that $H^* - 1 \leq \ceil{\log_s(4n / \BiasAbsolute)} \eqqcolon i^*$:
Indeed, otherwise all nodes go through round $i^*$ and a similar calculation as in \cref{thm:majority_clocks}'s proof yields
\begin{math}
     \abs{\TLoad{T_{i^*}}}
=    \abs{\STLoad{T_{i^*}}} \cdot s^{i^*}
=    \abs{\STLoad{0}} \cdot s^{i^*}
=    \alpha \cdot s^{i^*}
\geq 4n
\end{math}.
By an average argument as in \cref{thm:majority_clocks}'s proof, all nodes have absolute load $\geq 3$ after interaction $T_{i^*}$.
This implies that any node reaching hour $i^* + 1 < H^*$ has absolute load $\geq 3s$ after the load explosion and sets its finished bit, contradicting $H^*$'s choice ($i^*+1$ would be the maximal hour).

Thus, we have $H^* - 1 \leq i^* \leq h - 2$.
Let $u^*$ denote the initiator of interaction $T^*$ and remember the three possible reasons why $u^*$ could have set its finished or error bit (\crefrange{it:prf:stablmaj:reason:1}{it:prf:stablmaj:reason:3}).
\Cref{it:prf:stablmaj:reason:3} does not apply since no node skipped any hour.
\Cref{it:prf:stablmaj:reason:2} does not apply since the maximal hour is $H^* \leq i^* + 1 \leq h - 1$, so no node's phase counter overflows.
Thus, $u^*$ set its finished bit because of \cref{it:prf:stablmaj:reason:1}: it had absolute load at least $3s$ after its first interaction in hour $H^*$.
Then $u^*$ had absolute load at least $3$ after $T_{H^*-1}$ interactions (the end of round $H^* - 1$).
Together with \cref{it:prf:stablmaj:fast:2}, either all nodes had load at least $1$ or all nodes had load at most $-1$ after $T_{H^*-1}$ interactions.
In particular, all nodes have the same sign, which cannot change subsequently.
Since we already saw that load balancing happens only between nodes that experienced the same number of load explosions, \cref{lem:correct_opinion} implies that the nodes' sign is also the initial majority.

In summary, \whplong/, after $T_{H^* - 1}$ interactions, no error bit is set and all nodes \emph{forever} have the correct load sign.
This is still the case after interaction $T^*$.
Note that this does not imply $\TimeST/ \leq T^*$, since so far only one node finished and there is still a non-zero probability that some node skips an hour and, thus, sets the error bit after interaction $T^*$.
However, \whplong/, the finished bit spreads to all nodes within $\TimeINF/ = \LDAUOmicron{n \log n}$ interactions (the infection time, see \cref{lem:epidemic}).
Thus, by using \cref{lem:phaseclock} with a large enough constant $d_1$, we can ensure that, \whplong/, $\RoundLength(H^*) \geq \TimeINF/$, such that no node skips an hour before all finished bits are set.
Once all nodes are finished, the error bit cannot be set anymore, since all nodes have the same sign.

Combining everything above via a union bound, this yields that, \whplong/,
\begin{math}
     \TimeST/
\leq T^* + \TimeINF/
=    H^* \cdot \LDAUOmicron{n \log n} + \TimeINF/
=    \LDAUOmicron{n \log n \cdot \log_s(n / \BiasAbsolute)}
\end{math},
yielding the desired high-probability bound on the stabilization time.

Finally, we show that the stabilization time $\TimeST/$ of \ProtocolStableMaj{s} is $\ldauOmicron{n \log n \cdot \log_s(n / \BiasAbsolute)}$ in expectation.
To this end, observe that we know that, \whplong/, the stabilization time $\TimeST/$ is  $\ldauOmicron{n \log n \cdot \log_s(n / \BiasAbsolute)}$.
That is, for any constant $a > 0$, there is a constant $C > 0$ and appropriate values of the constant protocol parameters such that
\begin{equation}%
\label{eq1}
     \Prob{\TimeST/ \leq C \cdot n \log n \cdot \log_s(n / \BiasAbsolute)}
\geq 1 - n^{-a}
.
\end{equation}

To show that $\TimeST/$ is $\LDAUOmicron{n \log n \cdot \log_s(n / \BiasAbsolute)}$ in expectation, we show the following statement:
For some fixed $\eta$ (independent of the constant parameters of protocol \ProtocolStableMaj{s}), for each sufficiently large $n$, and for each configuration $\cC$ reachable from the initial configuration, the protocol stabilizes from $\cC$ within $n^{\eta}$ interactions in expectation.
Once this is shown, we can calculate
\begin{equation*}
\begin{aligned}
\Ex{\TimeST/}
&\leq
\phantom{+} \Ex{\TimeST/ |  \TimeST/ \leq C\cdot n \log n \cdot \log_s(n / \BiasAbsolute)}
\\&\phantom{{}\leq{}}
+ n^{-a}\cdot \Ex{\TimeST/ | \TimeST/ > C\cdot n \log n \cdot \log_s(n / \BiasAbsolute)}
\\&\leq
\phantom{+} C \cdot n \log n \cdot \log_s(n / \BiasAbsolute)
\\&\phantom{{}\leq{}}
+ n^{-a} \cdot \bigl( C \cdot n \log n \cdot \log_s(n / \BiasAbsolute) + n^{\eta} \bigr) 
\\&\leq
2 C \cdot n \log n \cdot \log_s(n/\BiasAbsolute)
.
\end{aligned}
\end{equation*}
The first inequality above follows from~\cref{eq1}.
The second inequality follows from the bound $n^{\eta}$ on the expected stabilization time from the configuration $\cC$ reached after the first $C \cdot n \log n \cdot \log_s(n / \BiasAbsolute)$ interactions.
Finally, the last inequality holds by taking $a = \eta$.

We use the following facts about the expected running time of some basic protocols:
\begin{enumerate}
\item Protocol \ProtocolBackupMaj/ stabilizes within $\ldauOmicron{n^2 \cdot \log n} \leq n^3$ (for sufficiently large $n$) interactions in expectation~\cite{DBLP:conf/icalp/MertziosNRS14}.

\item The one-way epidemic completes within $\LDAUOmicron{n \log n}$ interactions in expectation.

\item For each $K \geq 1$, the number of interactions required so that each node is the initiator of at least $K$ interactions is $\LDAUOmicron{K n \log n}$ in expectation (a simple consequence from the expected completion time $\LDAUOmicron{n \log n}$ of the coupon collector problem).
\end{enumerate}

Let $\cC$ be any non-stable configuration of the protocol \ProtocolStableMaj{s} that is reachable from the initial configuration.
We distinguish several configuration types with respect to $\cC$:
\begin{enumerate}[wide=0em, itemsep=0.618em]
\item\label[type]{it:majority_clocks_exact:case1} \emph{The phase clocks of all nodes have reached their limit of $h \cdot m$, viewing the clocks (for the purpose of this analysis but without actually altering anything in the protocol) as if there were kept running until reaching the limit:}
    Configuration $\cC$ is not stable, so either there is a node with the error flag raised, or all nodes have their finished flag raised but not all nodes have the same signs of their load.
    Otherwise the configuration is stable.

    In the former case (when there is a node with an error), one instance of one-way epidemic raises the error flag in all nodes in expected $\LDAUOmicron{n \log n}$ interactions, and then protocol \ProtocolStableMaj{s} stabilizes within additional expected $n^3$ interactions by completing \ProtocolBackupMaj/.

    Similarly, in the latter case, two consecutive one-way epidemics (the first one to make two nodes with different load signs meet and the second one to spread out the information about the error) and then the completion of the \ProtocolBackupMaj/ protocol are sufficient to stabilize \ProtocolStableMaj{s}.

    In both cases, the protocol \ProtocolStableMaj{s} stabilizes within additional $\LDAUOmicron{n^3}$ interactions  in expectation.

\item\label[type]{it:majority_clocks_exact:case2} \emph{There is at least one marked node (by protocol \ProtocolFormJuntaExt/), but there is still at least one node whose clock has not yet reached the limit of $h \cdot m$ (as above, we view the clocks as if running until the limit):}
    The marked node with the largest clock value increases its clock within one instance of one-way epidemics.
    Thus within at most $h \cdot m$ consecutive instances of one-way epidemics, one marked node reaches the clock limit.
    One additional one-way epidemic makes all clocks reach the limit.
    This takes $(h \cdot m + 1) \cdot \ldauOmicron{n \log n} = \ldauOmicron{n \log n\cdot \log_s(n)}$ interactions in expectation and takes us to a configuration of \cref{it:majority_clocks_exact:case1}.

\item \emph{No node is marked (by protocol \ProtocolFormJuntaExt/):}
    We consider two sub-cases:
    \begin{enumerate}
    \item\label[type]{it:majority_clocks_exact:case3a} \emph{All nodes in \ProtocolFormJuntaExt/ are inactive:}
        No node is ever marked and the phase clock never starts.
        Thus, \ProtocolStableMaj{s} stabilizes in $n^3$ additional interactions in expectation (via \ProtocolBackupMaj/).

    \item\label[type]{it:majority_clocks_exact:case3b} \emph{There is at least one active node in \ProtocolFormJuntaExt/:}
        If an active node is the initiator of an interaction, then it either increases its junta level or becomes inactive.
        Thus, when this node initiated $\lmax$ interactions, either it reached level $\lmax$ and got marked, or it has become inactive.
        So within $\lmax \cdot \ldauOmicron{n \log n} = \ldauOmicron{n\log n \cdot \log\log n}$ interactions in expectation (after each node initiated at least $\lmax$ interactions) we either reach a configuration with the first marked node (\cref{it:majority_clocks_exact:case2}) or a configuration with no marked node but only inactive nodes (\cref{it:majority_clocks_exact:case3a}).
    \end{enumerate}
\end{enumerate}

In summary, we see that \ProtocolStableMaj{s} stabilizes from any configuration $\cC$ reachable from the initial configuration within $\LDAUOmicron{n^3} \leq n^4$ (for sufficiently large $n$) interactions in expectation.
\end{proof}
}
{\section{Convergent Majority}%
\label{sec:convergent-majority}

In this section, we present and analyze the protocol \ProtocolConvergentMaj{s}, a hybrid majority protocol which converges efficiently.
The main idea of the protocol is that all nodes execute \ProtocolSimpleMaj{s,3}, which converges quickly.
However, there is a positive probability that it returns the wrong answer without detecting the error.
Therefore, every node switches its output to the backup protocol after a (polynomially) long time.
To determine that this time has passed, we use a simple approach based on counting the number of consecutive interactions with junta nodes.
Formally, we prove the following \lcnamecref{thm:convergent-majority}:
\begin{theorem}%
\label{thm:convergent-majority}
Let $s \in \set{2, 3, \dots, n}$.
Consider the majority problem for $n$ nodes with initial absolute bias $\BiasAbsolute \in \N$.
Protocol \ProtocolConvergentMaj{s} is exact and converges \whplong/ and in expectation in $\ldauOmicron{n \log n \cdot \log_s(n / \BiasAbsolute)}$ interactions.
The protocol uses $\ldauTheta{s + \log\log n}$ states per node.
\end{theorem}

We now describe the protocol's state space and its transition function (see also \cref{alg:convergent}).
Afterward, we give the proof of \cref{thm:convergent-majority}.
{\begin{algorithm}
\SetKw{KwLet}{let}
\SetKwProg{Algorithm}{}{}{}

\SetKw{KwLet}{let}
\SetKw{KwAnd}{and}
\SetKw{KwOr}{or}
\SetKwFunction{Exit}{exit}
\SetKwFunction{Balance}{balance}
\SetKwFunction{Multiply}{multiply}
\SetKwFunction{Wait}{wait}

\linesOff
\Algorithm{\ProtocolConvergentMaj{s}$(u,v)$}{
    \linesOn

	$\ProtocolBackupMaj/(u, v)$\;
\BlankLine
    \If{$\Count u < 600$}{
        \eIf{\PCmarked v}{
            $\Count u \gets \Count u + 1$\;
        }{
            $\Count u \gets 0$\;
        }
        \BlankLine\BlankLine

        \ProtocolSimpleMaj{s, 3}$(u,v)$\;
    }
}
\caption{
    Pseudocode for transition function of \ProtocolConvergentMaj{s} (initiator $u$ and responder $v$).
}
\label{alg:convergent}
\end{algorithm}
}

Nodes first execute a backup protocol \ProtocolBackupMaj/\footnote{%
    As before, in \cref{sec:majority_clocks_exact}, we use the $4$-state protocol from~\cite{DBLP:conf/icalp/MertziosNRS14} for this, which stabilizes in $\LDAUOmicron{n^2 \log n}$ interactions in expectation, implying a finite stabilization time and, thus, exactness.
}.
Additionally, each node $u$ executes protocol \ProtocolSimpleMaj{s, 3} as long as it did not encounter $600$ marked nodes in a row.
The number of such encounters is stored in a counter value $\Count{u} \in \set{0, 1, \dots, 600}$.
The value $600$ is chosen merely for convenience and has no special meaning.
It simply ensures that it takes a long time before a node permanently changes its output to that of the backup protocol (see next paragraph).

The output function maps the state of a node $u$ to a majority guess as follows: Use the output of the backup protocol if the phase counter of the phase clock is zero or if the counter \Count{u} has reached $600$.
Otherwise, use the output of protocol \ProtocolSimpleMaj{s, 3}.
Switching eventually to the backup protocol's solution ensures that – even if \ProtocolSimpleMaj{s, 3} fails – the protocol is exact.
Using the output of \ProtocolSimpleMaj{s, 3} in between (and switching to the backup protocol's solution only after a long time, when it is correct \whplong/) implies that, \whplong/, convergence (but not stability) is achieved fast.

\begin{proof}[Proof of \cref{thm:convergent-majority}]
We first show that our protocol \ProtocolConvergentMaj{s} is exact. This
follows easily by considering the following two cases:
\begin{enumerate}
\item The phase clock never starts (no node is selected into the underlying junta).
    Since in this case all counters remain zero forever, the output of the protocol \ProtocolConvergentMaj{s} equals the output of \ProtocolBackupMaj/, which is exact.

\item The phase clock starts (meaning the junta is not empty).
    The probability for a node $u$ to increase its $\Count{u}$ in the next interaction is at least $1/n^2$ ($u$ initiates the interaction with a marked node as responder).
    This happens $600$ times in a row with probability at least $1/n^{1200} > 0$ (a crude but sufficient bound).
    Thus, eventually all nodes $u$ reach $\Count{u} = 600$.
    From that point on the output of \ProtocolConvergentMaj{s} equals that of \ProtocolBackupMaj/, which is exact.
\end{enumerate}
This shows that protocol \ProtocolConvergentMaj{s} is exact. The bound on the
number of states per nodes follows also easily: Since \ProtocolBackupMaj/
requires only four states and the counters are bounded by the constant $600$,
the number of states per node is a constant factor times the number of states
required by \ProtocolSimpleMaj{s, 3}, which is $\LDAUTheta{s + \log\log n}$.

It remains to prove the convergence time bound for the
\ProtocolConvergentMaj{s} protocol. Assuming a non-empty junta (which, by
\cref{thm:junta}, holds \whplong/ after $\LDAUOmicron{n \log n}$ interactions),
we can derive the desired bound from the following observations.
\begin{enumerate}[(i)]
\item\label[obs]{it:convergent-majority:i} Once the junta is established, the output of protocol \ProtocolConvergentMaj{s} during the next $\poly_1(n)$ interactions equals, \whplong/, that of \ProtocolSimpleMaj{s, 3}.
    That protocol converges \whplong/ in at most $\LDAUOmicron{n \log n \cdot \log_s(n / \BiasAbsolute)}$ interactions.

\item\label[obs]{it:convergent-majority:ii} After $\poly_1(n)$ interactions, the nodes' outputs start switching gradually to the output of \ProtocolBackupMaj/, which stabilizes in $\LDAUOmicron{n^2 \log n}$ interactions.
    After $\poly_2(n) > \poly_1(n)$ interactions, \whplong/ all nodes have switched their output to \ProtocolBackupMaj/.

\end{enumerate}
By choosing parameters such that $\poly_1(n) = n^3$, the switch to the backup protocol happens only when it has, \whplong/, stabilized.
Thus, \whplong/, the subprotocol \ProtocolSimpleMaj{s, 3} converges to the correct outcome within $\LDAUOmicron{n \log n \cdot \log_s(n / \BiasAbsolute)}$ interactions and, when the nodes start switching their output to \ProtocolBackupMaj/, that subprotocol also has the correct output.
Together, this implies the desired bound on the convergence time.

To see \cref{it:convergent-majority:i}, note that, once there is a non-empty junta of size at most $n^{\chosenoneminusxi/}$ (which happens \whplong/ in $\LDAUOmicron{n \log n}$ interactions by \cref{thm:junta}), the probability that a node samples a junta node $600$ times in a row is at most ${(n^{\chosenoneminusxi/}/n)}^{600} = n^{-12}$.
Using a union bound, \whplong/ no node reaches counter value $600$ (and switches to the backup protocol) before $\LDAUTheta{n^3}$ interactions.
\Cref{it:convergent-majority:ii} follows by a simple Markov bound applied to the expected number of interactions a node requires to switch its output back to the backup protocol (which is upper bounded by $\LDAUOmicron{n^{600}}$) together with a union bound over all nodes.

It remains to bound the expected convergence time $\TimeC/$ of \ProtocolConvergentMaj{s}.
For this, using the same argument as in the proof of \cref{thm:majority_clocks_exact}, it is sufficient to show the following statement:
For some fixed $\eta$, for each sufficiently large $n$, and for each configuration $\cC$ reachable from the initial configuration, the protocol stabilizes from $\cC$ within $n^{\eta}$ interactions in expectation.
Once this is shown, the same calculation as for \cref{thm:majority_clocks_exact} yields
\begin{math}
     \Ex{\TimeC/}
\leq 2C \cdot n \log n \cdot \log_s(n/\BiasAbsolute)
\end{math}.
To this end, we proceed as in the proof of \cref{thm:majority_clocks_exact} and distinguish the following configuration types:
\begin{enumerate}[wide=0em, itemsep=0.618em]
\item\label[type]{it:convergent-majority:conftype:case1} \emph{All nodes have switched to the backup protocol:}
    Then \ProtocolConvergentMaj{s} stabilizes in $n^3$ additional interactions in expectation (via \ProtocolBackupMaj/).

\item\label[type]{it:convergent-majority:conftype:case2} \emph{Some node has not yet switched to the backup protocol:}
    The time until all nodes switch to the backup protocol is dominated by the sum of $n$ geometrically distributed random variables with parameter $\geq 1 / n^{1200}$ (see the exactness proof above).
    Thus, there is a constant $\eta$ such that all nodes switch their output to \ProtocolBackupMaj/ after at most $n^{\eta-1}$ many interactions in expectation.
    This takes us to a configurations of \cref{it:convergent-majority:conftype:case1}.
\end{enumerate}

In summary, wee see that \ProtocolConvergentMaj{s} converges from any configuration $\cC$ reachable from the initial configuration within $n^{\eta-1} + n^3 \leq n^{\eta}$ (for sufficiently large $n$) interactions in expectation.
\end{proof}
}
{\section{A Note on Uniformity}%
\label{sec:uniform-population-protocols}

Uniformity in population protocols means that a single algorithm is designed to work for populations of any size.
In particular, nodes have no information on the population size $n$.
Protocols where nodes are restricted to a constant number of states are always uniform.
But most of the newer protocols allow for a super-constant number of states and use some upper bounds on $n$, so they are not uniform.
In particular, protocols that stop their computation once a counter reaches a value of $\polylog(n)$ fall into this category of non-uniform protocols.
This \lcnamecref{sec:uniform-population-protocols} presents a uniform population protocol for majority.

\paragraph{Uniform Population Model}
To study \emph{uniform} population protocols whose state requirements increase with the population size $n$, the original model – which considers nodes as finite-state machines (FSM), see \cref{sec:model} – turns out to be inadequate.
Indeed, if each node is an FSM with a state space of size $f(n)$ for a non-trivial function $f$, then the nodes and, thus, the protocol inherently depend on $n$ and cannot be simply \enquote{deployed} in a population of different size.

\Textcite{DBLP:conf/podc/DotyE19} introduce a generalized population model that is better suited for this scenario and which we adopt in the remainder of this \lcnamecref{sec:uniform-population-protocols}.
In their model, each node is represented by a $2$-tape deterministic Turing machine (TM).
We assume that both tapes are infinite to the left and right and that the origin is marked by a special \emph{origin symbol}.
Tape~1 (read-only) is called the \emph{input tape} and tape~2 (read-write) the \emph{working tape}.
One two-way infinite working tape is sufficient for us since it allows maintaining two unbounded variables, as required in our protocol.
For protocols with more unbounded variables, similarly to~\cite{DBLP:journals/corr/abs-1805-04832} one can use a TM with as many (one-sided infinite) input/working tapes as there are unbounded variables (whose number must not depend on $n$).

At the beginning of any interaction, a node's working tape is identical to its working tape at the end of the previous interaction.
Whenever two nodes interact, they copy each other's working tape onto their own input tape and restart their TM by entering a start TM-state (which then computes a new state, updates the node's working tape, and halts).
We define the number of states used during a protocol execution as $\abs{\Sigma}^s$, where $\Sigma$ is the (binary) tape alphabet and $s$ is the maximum number of tape cells written by any node during the execution.

Having the above formal model in mind, our description sticks with the standard population protocol terminology.
In particular, we assume a suitable encoding of the nodes' states using the alphabet $\Sigma$ and simply identify the content of a node's working tape with its state.
An important implication of the model is that nodes might now use an unbounded number of states (write an unbounded number of cells on the working tape).
However, in our uniform majority protocol, the number of used states is finite with probability $1$ and $\ldauOmicron{\log n \cdot \log\log n}$ in the population size $n$ \whplong/.

\paragraph{Uniform Majority}
One of the rare examples of a uniform protocol whose state requirements increase with $n$ is the junta protocol from~\cite{DBLP:conf/soda/GasieniecS18}; we refer to it as \ProtocolFormJunta/.
Observe that our protocol \ProtocolFormJuntaExt/ is not uniform, as nodes need to know $\lmax = \floor{\log\log n} - 3$ in order to mark themselves (see \cref{sec:junta_creation}).
See below for a brief description of \ProtocolFormJunta/.

Since our majority protocols from the previous sections use the non-uniform junta \ProtocolFormJuntaExt/, none of them is uniform.
In fact, to the best of our knowledge, until now there was no exact, uniform majority protocol that would stabilize \whplong/ in $n^{2 - \LDAUOmega{1}}$ interactions.
The following \lcnamecref{thm:majority:stabilizing:uniform} shows that we get such a uniform majority protocol by applying slight modifications to protocol \ProtocolStableMaj{s}.
\begin{theorem}%
\label{thm:majority:stabilizing:uniform}
Let $s \in \N \setminus \set{1}$ be a constant.
Consider the majority problem for $n$ nodes with initial absolute bias $\BiasAbsolute \in \N$.
Protocol \ProtocolUniformMaj{s} is an exact and uniform variant of \ProtocolStableMaj{s}.
\Whplong/ and in expectation, it stabilizes in $\ldauOmicron{n \log n \cdot \log_s(n / \BiasAbsolute)}$ interactions.
While the number of used states can be arbitrarily high with non-zero probability, \whplong/ it uses only $s \cdot \ldauOmicron{\log_s(n / \BiasAbsolute) \cdot \log\log n}$ states per node.
\end{theorem}
Note that, in order for \ProtocolUniformMaj{s} to be uniform, the parameter $s$ must be constant (i.e., $s$ may not depend on $n$).

Protocol \ProtocolUniformMaj{s} is identical to protocol \ProtocolStableMaj{s} with the following changes:
\begin{enumerate}[noitemsep]
\item It uses subprotocol $\ProtocolSimpleMaj{s, \infty}$ instead of $\ProtocolSimpleMaj{s, \ceil{\log_s(4n)} + 2}$.
    In particular, the phase clock \ProtocolPhaseClock{\infty} is used, which cannot overflow. Thus, nodes always know their exact hour.

\item The phase clock uses \ProtocolFormJunta/ instead of \ProtocolFormJuntaExt/.
\end{enumerate}
Using the original junta instead of ours has the drawback that nodes must remember their level from the junta calculation indefinitely.
However, since protocol \ProtocolFormJuntaExt/ is inherently non-uniform, this seems unavoidable when aiming for a uniform protocol.

For the sake of completeness, we give a brief description of
\ProtocolFormJunta/, using slightly different wording and notation than
in~\cite{DBLP:conf/soda/GasieniecS18}, in order to fit it into our
framework\footnote{%
    Technically, the described protocol differs slightly from the one
    in~\cite{DBLP:conf/soda/GasieniecS18}: The first interaction of a node is
    slightly changed – as described in \cref{sec:junta-calculation} – in order
    to enable us to prove a lower bound on the maximum level reached by any
    node. This property is not required for the uniform protocol, so one could
    use the original junta protocol from~\cite{DBLP:conf/soda/GasieniecS18}.
}. We also describe how the phase clock is adapted to the changed junta
protocol. Afterward, we give the proof of
\cref{thm:majority:stabilizing:uniform}.

\paragraph{Description of \ProtocolUniformMaj{s}}
As our junta protocol \ProtocolFormJuntaExt/, protocol \ProtocolFormJunta/ is based on the level calculation described in \cref{sec:level-calculation}.
Recall that the level calculation uses a level $l$, an activity bit $a$, and the transition function described by \cref{eq:level-calculation}.
In addition to the level $l$ and activity bit $a$, each node stores a \emph{marker bit} $b \in \set{0, 1}$ (indicating whether the node assumes to be in the junta or not) and a \emph{defeated bit} $d \in \set{0, 1}$.
Initially, all nodes have $b = 0$ and $d = 0$.
A node that just became inactive at a level $l \geq 1$ sets $b$ to $1$.
If an inactive node at level $l \geq 1$ encounters a node on a higher level, it becomes \emph{defeated}: it sets $d$ to $1$, $b$ to $0$, and will from now on simply adopt the larger level during any interaction (not changing any of its other state values related to the junta).
If encountered by another node in the level calculation, a defeated node is treated as if it were in state $(0, 0)$, independent of its actual level counter $l$.

\paragraph{Phase Clocks on different Levels \& Reset}
Compared to \ProtocolFormJuntaExt/, any (inactive) node starts with the belief
of being in the junta until it becomes defeated by a node from a higher level.
This ensures that the junta is never empty. However, when used in the phase
clock protocol, there will be a large number of nodes in the junta for the
first few interactions (until lower-level nodes become defeated), causing the
phase clock to run too fast.

To avoid problems in the protocol relying on the synchronization of the phase
clock, nodes now also use the level (from the junta protocol) in the phase
clock protocol. This basically results in multiple phase clocks running on
different levels. When a node running a phase clock on level $l$ encounters a
node running a phase clock on a higher level $l'$, it resets its phase counter
to zero (and – by the junta protocol – updates its level to $l'$). This
\emph{reset} also triggers a reset of the protocol using the phase clock. In
our case this is the majority protocol \ProtocolStableMaj{s}. For a node $u$
this entails a reset of the bits $\Finished{u}$ and $\Error{u}$ to $0$, and a
reset of the load value from \ProtocolSimpleMaj{s, \infty} to $\pm 1$,
depending on the original opinion of $u$. This idea of phase clocks running on
different levels and a corresponding reset was first proposed and used
in~\cite{DBLP:conf/soda/GasieniecS18} for the case of leader election.

\begin{proof}[Proof of \cref{thm:majority:stabilizing:uniform}]
First note that the protocol requires no knowledge of $n$, meaning that it is uniform.
The remaining proof is similar to that of \cref{thm:majority_clocks_exact}.
In fact, having unbounded phase counters (which avoid overflows in the phase clock) and the guarantee from \ProtocolFormJunta/ that the junta is never empty simplify the argumentation considerably.
Also note that the exact phase counters guarantee that any load balancing action is \emph{always} guaranteed to be done only between nodes in the same hour.

Let $\TimeST/$ denote the stabilization time of protocol \ProtocolUniformMaj{s}.
To proof exactness, remember the three cases from the exactness proof of \cref{thm:majority_clocks_exact}.
The first two cases are trivial:
\Cref{it:majority_clocks_exact:exactcases:i} (phase clock does not start) cannot occur, since we use \ProtocolFormJunta/.
\Cref{it:majority_clocks_exact:exactcases:ii} (phase clock starts and some error bit is set) is identical, since with probability $1$ eventually all nodes set their error bit and use the output of the backup protocol.
For \cref{it:majority_clocks_exact:exactcases:iii} (phase clock starts and no error bits are ever set) we again first show that all nodes finish with probability $1$.
However, in \cref{thm:majority_clocks_exact} this was proven via the overflowing phase counters, which cannot happen for $h = \infty$.
Thus, we use a different argument:
For the sake of a contradiction assume no node finishes (if one node finishes, all nodes finish eventually).
Since the phase clock runs and no error ever occurs, all nodes reach any hour $i \in \N$.
Consider an interaction $t$ when all nodes are in hour at least $\iota \coloneqq \ceil{\log_s(s \cdot 3n / \BiasAbsolute)}$.
Since load balancing is only performed between nodes in the same hour and no node ever skips an hour (or an error would occur), \cref{lem:correct_opinion} gives $\Psi(t) = \Psi(0)$.
But then, similar to previous arguments, our choice of $\iota$ ensures that $\abs{\TLoad{t}} \geq s \cdot 3n$.
So there would be a node $u$ with absolute load at least $3s$.
This yields the desired contradiction, since $u$ would have set its finished bit at the beginning of its current hour.
Once we know that all node finish with probability $1$ in this case, the exactness follows again as in \cref{it:majority_clocks_exact:exactcases:iii} in the exactness proof of \cref{thm:majority_clocks_exact}.

To prove that, \whplong/, we have $\TimeST/ = \ldauOmicron{n \log n \cdot \log_s(n / \BiasAbsolute)}$, we use the same argumentation as for the corresponding part in the proof of \cref{thm:majority_clocks_exact}, just slightly simplified since nodes now store their exact phase counters.
Basically, we again take a union bound over the first $i^* \coloneqq \ceil{\log_s(4n / \alpha)}$ rounds and get the same three properties as in the proof of \cref{thm:majority_clocks_exact}:
\begin{enumerate*}[(1)]
\item $T^* < \infty$ (as we have shown above for the exactness) for the interaction $T^*$ when the first finish or error bit is set.
\item The first $i^*$ rounds have length $\LDAUOmega{n \log n}$ and stretch $\LDAUOmicron{n \log n}$.
\item The load discrepancy is at most $2$ at the end of each of the first $i^*$ rounds.
\end{enumerate*}
With these properties, the remaining argumentation from \cref{thm:majority_clocks_exact}'s proof goes through.

The proof for the bound on the expected stabilization time is also identical, by noting that all nodes complete protocol \ProtocolFormJunta/ in expected $\LdauOmicron{n \log n}$ interactions~\cite{DBLP:conf/soda/GasieniecS18}.

For the bound on the number of states, note that replacing the phase clock's junta algorithm \ProtocolFormJuntaExt/ by \ProtocolFormJunta/ increases the required number of states by a \emph{factor} of $\LDAUTheta{\log\log n}$ (instead of an \emph{additive} term), since we now need to store the level indefinitely.
Now, above we saw that, \whplong/, all nodes finish after at most $\ldauOmicron{\log_s(n / \BiasAbsolute)}$ rounds.
Thus, \whplong/, no phase counter is larger than $\ldauOmicron{\log_s(n / \BiasAbsolute)}$ in absolute value.
Finally, there is a factor of $(3s+1) \cdot \LDAUTheta{1} = s \cdot \LDAUTheta{1}$ for the load values and the bits $\Finished{u}$ and $\Error{u}$, yielding the desired bound.
\end{proof}
}
{\section{Conclusions \& Future Work}

We analyzed three similar variants of a population protocol for the majority problem: \ProtocolSimpleMaj{s, 3}, \ProtocolConvergentMaj{s}, and \ProtocolUniformMaj{s}.
All of them based on the so-called \emph{doubling and cancellation approach}.
They feature a parameter $s$ that allows for a trade-off between runtime and memory per node.

A natural open question is to improve the bounds we provide.
In particular, for $s = \log\log n$ our protocol \ProtocolStableMaj{s} has stabilization time $\ldauomicron{n \cdot {(\log n)}^2}$ while using $\LDAUOmicron{\polylog n}$ states.
There is (to the best of our knowledge) one other result that also achieves this guarantee~\cite{DBLP:conf/wdag/BerenbrinkEFKKR18}.
While it does not feature a trade-off capability, it comes with a better stabilization time.
It seems non-trivial but also not impossible to combine our trade-off result with the improved stabilization time.
Also, it would be interesting whether it is possible to derive parameterized \emph{lower bounds} in which one can similarly see the effect on the running time of increasing or decreasing the number of states per node.

Another open research question for population protocols deals with the phase clock introduced in~\cite{DBLP:conf/soda/GasieniecS18}.
It is unclear whether one can derive a similar phase clock that requires only a constant number of states and still synchronizes the population for a polynomial number of interactions \whplong/.
If it exists, such a phase clock could be used to devise constant-state (majority) protocols that converge in polylogarithmic time \whplong/.

Our results formally show that lower bounds for the stabilization time can be bypassed by considering the convergence time.
Unfortunately, there are currently no strong lower bounds regarding the convergence time.
As convergence time might be considered the more practical runtime notion, finding such lower bounds and tightening the corresponding upper bounds should be deemed a worthy but challenging task.
}

\printbibliography

\appendix
{\section{Probabilistic Tools}%
\label{app:probtools}

\begin{lemma}[Chernoff Bounds]%
\label{lem:chernoff:mult}
Let $n \in \N$ and consider a sequence ${(X_i)}_{i \in \intcc{n}}$ of mutually
independent binary random variables. Define $X \coloneqq \sum_{i \in \intcc{n}}
X_i$ and let $\mu_U, \mu_L \geq 0$ be such that $\mu_L \leq \Ex{X} \leq \mu_U$.
The following inequalities hold for any $\delta \geq 0$ and $\phi \geq 6\mu_U$:
\begin{align}
      \Prob{X \leq (1 - \delta) \cdot \mu_L}
&\leq e^{-\frac{\delta^2 \cdot \mu_L}{2}},
\label{eq:chernoff:mult:a}
\\
      \Prob{X \geq (1 + \delta) \cdot \mu_U}
&\leq e^{-\frac{\delta^2 \cdot \mu_U}{2 + \delta}}, \qquad\text{and}
\label{eq:chernoff:mult:b}
\\
      \Prob{X \geq \phi}
&\leq 2^{-\phi}
\label{eq:chernoff:mult:c}
.
\end{align}
\end{lemma}
Let $\mu \coloneqq \Ex{X}$. We often use the following simplified Chernoff
bounds:
\begin{align}
      \Prob{X \leq (1 - \delta) \cdot \mu}
&\leq n^{-a}
\label{eq:chernoff:mult:d}
\\
      \Prob{X \geq \max\set{13a \cdot \log n, (1 + \delta) \cdot \mu}}
&\leq n^{-a} \qquad\text{and}
\label{eq:chernoff:mult:e}
,
\end{align}
where $a \geq 0$ is an arbitrary constant and $\delta \coloneqq \sqrt{3a \cdot
\log(n) / \mu}$. For convenience, we sometimes combine both bounds into
\begin{equation}\label{eq:chernoff:mult:combined}
     \Prob{\abs{X - \mu} \geq \max\set{13a \cdot \log n, \delta \cdot \mu}}
\leq 2n^{-a}
.
\end{equation}

\section{Auxiliary Protocols: Phase Clock}%
\label{app:phase-clock}

This section shows how \cref{lem:phaseclock} follows from the following
technical lemma from~\cite{DBLP:conf/soda/GasieniecS18}. We paraphrase the
lemma slightly in order to make the dependencies on the involved constants more
explicit.
\begin{lemma}[{\cite[Lemma~3.7]{DBLP:conf/soda/GasieniecS18}}]%
\label{lem:phaseclock:helper}
Let $a, d > 0$ be constants and assume $n$ to be sufficiently large with
respect to them. There is a constant $K > 0$ such that the following holds: Let
$p_{\max}$ denote the maximum and $p_{\min}$ the minimum phase counter after an
interaction $t \in \N$. Assume $p_{\max} - p_{\min} \leq 2K$. With probability
at least $1 - n^{-a}$, there is a $t' > t + d \cdot n \log n$ such that:
\begin{enumerate}
\item

    $t'$ is the first interaction after which the maximum phase counter is
    $p_{\max} + K$.

\item

    After interaction $t'$, all nodes have a phase counter value of at least
    $p_{\max}$.

\end{enumerate}
\end{lemma}


With this, we are ready to restate and prove \cref{lem:phaseclock}.
\lemphaseclocks*
\begin{proof}
The lower bound on $\RoundLength(i)$ follows via an induction over $i$ by applying \cref{lem:phaseclock:helper} with $d = d_1$ and by setting the phase clock parameter $m$ to $3K$.
For the upper bound on the stretch, note that the one-way epidemic (cf.~\cref{lem:epidemic}) implies that, \whplong/, the maximum phase counter increases within $\LDAUOmicron{n \log n}$ rounds (when a marked node finally sees the maximum phase counter).
Thus, \whplong/, it takes at most $m \cdot \LDAUOmicron{n \log n} = \LDAUOmicron{n \log n}$ interactions for a node to leave a given round.
\end{proof}
}
\end{document}